\documentclass[11pt]{article}
\usepackage{graphicx} %
\usepackage[margin=1in]{geometry}
\usepackage{import}

\usepackage{mystyle.tex}
\usepackage{papermacros.tex}

\newcommand{\remove}[1]{}

\title{Shuffling Cards When You Are of Very Little Brain: \\Low Memory Generation of Permutations\thanks{Research supported in part by grants from the Israel Science Foundation (no.\  2686/20),  by the Simons Foundation Collaboration on the Theory of Algorithmic Fairness and by the Israeli Council for Higher Education (CHE) via the Weizmann Data Science Research Center.}
}
\author{Boaz Menuhin\thanks{
		Some of this work was accomplished when the author was at the Department of Computer Science and Applied Mathematics, Weizmann Institute of Science, Rehovot, Israel.
		Email: {boaz.menuhin@gmail.com}.
	}
	\and
	Moni Naor\thanks{
		Department of Computer Science and Applied Mathematics, Weizmann Institute of Science, Rehovot, Israel.
		Incumbent of the Judith Kleeman Professorial Chair.
		Email: {moni.naor@weizmann.ac.il}.}
}

\begin{document}
\maketitle

\begin{abstract} 
How can we generate a permutation of the numbers $1$ through $n$ so that it is hard to guess the next element given the history so far? The twist is that the generator of the permutation (the ``Dealer") has limited memory, while the ``Guesser" has unlimited memory. 
With unbounded memory (actually~$n$ bits suffice), the Dealer can generate a truly random permutation where~$\ln n$ is the expected number of correct guesses.

Our main results establish tight bounds for the relationship between the guessing probability and the memory $m$ required to generate the permutation.
We suggest a method for an $m$-bit Dealer that operates in constant time per turn, and any Guesser can pick correctly only $O(n/m+\log m)$ cards in expectation. The method is fully transparent, requiring no hidden information from the Dealer (i.e., it is "open book" or "whitebox").

We show that this bound is the best possible, even with secret memory. Specifically, for
any $m$-bit Dealer, there is a (computationally powerful) guesser that achieves $\Omega(n/m+\log m)$ correct guesses in expectation. We point out that the assumption that the Guesser is computationally powerful is necessary: under cryptographic assumptions, there exists a low-memory Dealer that can fool any computationally bounded guesser. 

We also give an $O(n)$ bit memory Dealer that generates perfectly random permutations and operates in constant time per turn. 
 
\end{abstract}

\section{Introduction}
The question of generating random permutations has received significant attention, dating back to at least the work of Fisher and Yates in the 1930s (see Section 3.4.2 in Knuth's~\cite{Knuth98}). In this work we concentrate on generating permutations using a small amount of memory while ensuring that it is difficult to predict the next value based on the previously generated values.
	
The following game was considered by Menuhin and Naor~\cite{MenuhinN22}:
A card guessing game is played between two players, ``Guesser" and ``Dealer".
At the beginning of the game, the Dealer holds a deck of $n$ distinct cards (labeled $1, ..., n$).
For~$n$ turns, the Dealer draws a card from the deck, the Guesser guesses which card was drawn, and then the card is discarded from the deck.
The Guesser receives a point for each correctly guessed card.
	
Menuhin and Naor considered the asymmetric case where the Dealer remembers everything, and the Guesser has limited memory and derived tight bounds in this case. 
In this work we consider what happens when the shoe is on the other foot and the Dealer is limited in space, while the Guesser has plenty of memory. 
Say that the Dealer has only~$m$ bits of memory, and wants to pick a permutation that is unpredictable by any Guesser, that has no limitation on the number of bits it can store or on its computational power? We know that when both players have plenty of memory, then the expected number of correct guesses is $H_n \approx \ln n$.

In our model, the Dealer has a limited number of long-lived random bits, in the sense that storing them is part of its memory, as well as ``on-the-fly" randomness that is available at every round. 
Note that with~$O(n \log n)$ long-lived bits, the Dealer can generate a truly random permutation.	
As Knuth showed~(see Algorithm P in Section 3.4.2 in~\cite{Knuth98}), we can drop the longevity requirement by maintaining an ordered list of available cards, pick one at random and discard. This algorithm, which is attributed to Fisher and Yates requires~$O(n \log n)$ bits of memory.
It is also possible to get a {\em perfectly random permutation} with only $n$ bits of memory, using a bitmap Dealer (see below), where what the generator maintains is the subset of cards that was used so far.  

If the Dealer has only $m$ bits of memory, then if we partition the deck into $n/m$ parts and for each one generate a perfect permutation using the same storage of $m$ bits (that is $n/m$ perfect permutations on sets of size $m$). For each such small permutation, the expected number of guesses a perfect Guesser makes is $\ln m$, and thus $(n/m) \ln m$ altogether. 

The natural question is whether this algorithm is the best possible in terms of its unpredictability or whether it can be improved.  
After some pondering the reader may conclude that the best strategy such a Dealer can have is pick the next card from a set of $m$ cards at random (in the above description, towards the end of each epoch the number of possibilities diminishes), and the question is how to ensure that there is such a set available for as many rounds as possible especially given the tight memroy requirements. This is indeed true, and our paper proves it and, furthermore, shows that this is the best possible. 

\paragraph{Our Contributions:}
Our main results are tight bounds for the relationship between the guessing probability and the memory required to generate the permutation.

For the case where the Dealer is frugal and is willing to allocate only~$m=o(n)$ bits of memory,
we found a method that makes any Guesser pick correctly only $O(n/m+\log m)$ cards in expectation.
The method does not require any secrecy from the Dealer, i.e.\ it is ``open book" in the language of Magen and Naor~\cite{MagenN22}  or ``whitebox" in terms of Ajtai et al.~\cite{AjtaiBJSSWZ22}).
Furthermore, this method is computationally efficient, where each card draw takes a constant time in the worst case and overall~$O(n)$ time.

On the other hand, we show that this bound is the best possible, even for Dealers with secret memory. For any $m$ bit Dealer there is a (computationally powerful) Guesser, which we call the {\em myopic optimizing Guesser}, that makes $\Omega(n/m+\log m)$ correct guesses in expectation. 
	
Note that the Guesser from the lower bound is {\em not} efficient, and this is not a coincidence: we show that assuming that the Guesser is polynomial time and that one-way functions exist, then it is possible to generate a permutation that is indistinguishable from random (and hence any computationally efficient Dealer can guess at most $\ln n$ cards correctly in expectation) using an amount of memory needed to store a key to a pseudo-random function (PRF); See \Cref{one-way}.

For the high memory regime, we introduce a Dealer that produces a perfectly random permutation, that is, one against which any Guesser scores at most~$\ln n$ points in expectation, runs in constant time, and requires a linear amount of bits.
This ultimately closes the gap between the \bitmapDealer that uses~$n$ bits of memory and runs in super-constant time (even amortized), and Knuth's implementation of Fisher and Yates's algorithm that runs in constant time and requires~$\Omega(n \log n)$ bits of memory.
To the best of the authors' knowledge, this is the first suggestion to achieve this.

\subsection{Related Work}

\paragraph{Generating Permutations and Cards Shuffling}
Due to its fundamental nature, the problem of generating a random permutation has been widely treated and analyzed from different angles and has a relatively long history.
The computational aspect of generating a permutation got a dedicated treatment in Knuth's~\cite{Knuth98}, where he analyzed a method by Fisher and Yates from the 1930s.
In an iconic paper among many, Bayer and Diaconis~\cite{BayerDiaconis92} analyzed the amount of Rifle-Shuffles required to achieve a close to uniform distribution of permutations.
Shuffles are studied cryptographically, as in Morris and Rogaway's~\cite{MorrisR14} construction (see Section~\ref{one-way}).
Permutations are essential, and thus, so is the study of the computational resources required to produce them.

\paragraph{The Missing Item Problem:}
In the Missing Item Finding problem (MIF), a stream of $r$ arbitrary elements 
$(e_1, ..., e_r) \in [n]^r$ is presented to an observer, which has to come up with some element $x \in [n]$, that has not appeared yet in the stream so far.
Stoekl~\cite{stoeckl2023streaming} and Magen~\cite{magen2024missing} thoroughly analyzed upper and lower bounds for various computational variants of the MIF problem, and in particular the one in which the generator's internal are ``open-book''.
In each turn, our Dealer faces a similar challenge, where it must produce a never-seen-before element.

\paragraph{Mirror Games:} Mirror games were invented by Garg and Schnieder~\cite{GargS19}. In this game, Alice and Bob take turns (with Alice playing first) in declaring numbers from the set $\{1,2, \dots, 2n\}$. If a player picks a number that was previously played, that player loses and the other player wins. If all numbers are declared without repetition, the result is a draw. Bob has a simple mirror strategy that assures he won't lose and requires no memory. On the other hand, Garg and Schnieder showed that every deterministic Alice needs memory of size linear in $n$ in order to secure a draw. Regarding probabilistic strategies, with sufficiently many secret bits or using cryptographic assumptions, there is a strategy for Alice that enables her to achieve a draw in the game w.h.p.\ using only polylog bits of memory~\cite{GargS19,Feige,MenuhinN22}.
The requirement for secret bits is crucial~\cite{MagenN22}.

\paragraph{Dynamically Changing Distributions}
The task of sampling a random variable from a given discrete distribution has been ultimately solved by the famous Alias method.
The case where the distribution changes across time has been studied by Rajasekaran and Ross~\cite{rajasekaran_fast_1993}, who suggested several rejection-based algorithms that work well for some families of dynamic distributions.
Their work followed a fruitful line of work for handling more classes of dynamically changing distributions, where Matias, Vitter and Ni~\cite{MatiasVN93} suggested a general solution that works for any dynamic distribution, albeit one that runs at a super constant time.
Inspired by their work, Hagerup, Mehlhorn and Munro~\cite{hagerup1993} introduced two algorithms for different classes of dynamic distributions.
While many of these works reached expected constant time (or near it), none reached worst case constant time, and they have neglected the space aspect; thus, to the best of our knowledge, none of our results immediately follow from this branch.

Sampling marbles from urns, with and without replacements, is a dynamically changing distributions of special interest in simulation of interactions between agents.
Berenbrink, Hammer, Kaaser, Meyer, Penschuck and Tran~\cite{berenbrink_simulation20} presented a memory efficient rejection-based algorithm for this task, especially designed for bulk operations.
One of our machineries addresses this distribution, as well as a broader family of distributions, and runs faster. However, it also consumes much more memory, and is not intended for bulk operation.

\remove{
\paragraph{Generative streaming algorithms}
Many streaming algorithms handle the task of analyzing or responding to a given stream.
Our Dealer's algorithms are generative in the sense that they are independent of a given input and focuses on the task of generating an object (in our case, a permutation) in a streaming fashion.
An interesting property of this task is that the domain of generation is widely used and is often examined by its randomness.
}

\paragraph{Generating and Storing Permutations} 
\label{storing_permutations}
The randomness of a Guesser-agnostic Dealer dictates the drawing order of the cards and, thus, describes a permutation.
In the other direction, given a succinct permutation data structure that can be initialized with random bits, we can generate a random drawing order.

Storing a permutation succinctly aims at representing a permutation using space as close to the information-theoretic lower bound as possible while enabling a quick (adaptive) evaluation of arbitrary elements. Such schemes are significant for Proof-of-Space schemes~\cite{BonehCohen2023} and database optimization~\cite{MunroRajeev2012}.

\paragraph{Circuit Depth of Generating Permutations:} It is possible to generate random permutations in parallel. In particular, it can be done in the class $AC^0$, polynomial size, constant depth unbounded fan-in AND/OR gates, but not in the class $NC^0$, circuits where every output bit depends on a constant number of input bits~\cite{MatiasV91,Hagerup91,Viola12,Viola20}.

\subsection{Results and Structure}

As mentioned, our primary results are smooth bounds relating the amount of memory that the Dealer poses and the predictability of the permutations they can produce.

\paragraph{Low Memory Dealer:}
With~$m$ bits of memory, we aim to have~$\Omega(m)$ cards to choose from in every turn.
Our Dealer is based on the idea that it is easier to track the progress of many small mini-decks than the deck in its entirety.
After splitting the deck, in each turn, the Dealer chooses a mini-deck to draw a card from and draws the top card of that mini-deck. 
This already results in a huge reduction in the required memory.
If the mini-decks are close to each other, then we can store their distances from some central point and thus achieve additional savings.
So we want the mini-decks to progress somewhat evenly without giving too much advantage to the Guesser; this is exactly the kind of problem addressed by allocation schemes in the balls-into-bins model (See~\Cref{sec_preliminaries:sec_balls_into_bins} and a survey by Wieder~\cite{Wieder2017}).

Our Dealer utilizes the~\adaptiveTerm allocation scheme introduced by Berenbrink, Khodamoradi, Sauerwald, and Stauffer~\cite{Berenbrink_Khodamoradi_Sauerwald_Stauffer_2013}, which can be roughly described as ``sample mini-decks until finding one that is not much ahead of the average''.
As we will show (\Cref{sec_adaptive:claim_number_of_minidecks}), this scheme can be implemented by a Dealer with~$m$ bits of memory can track~$\Theta(m)$ mini-decks. But how well does the Dealer play?
It turns out that it is pretty good in the sense that there are $\Theta(m)$ cards that can be played at every round. Therefore, our Dealer scores~$O(n/m + \ln m)$ points in expectation against any Guesser (\Cref{sec_adaptive:sec_predictable}.)
Furthermore, this Dealer is efficient and runs in constant time in expectation.

We prove that the Dealer has these properties by walking the path paved by Berenbrink et al. in~\cite{Berenbrink_Khodamoradi_Sauerwald_Stauffer_2013}.
As in other proofs on allocation schemes, our main tool is a potential function.
We relate the value of the potential function to the set of mini-decks that we can draw from and show that we expect to have sufficiently many such mini-decks. 
This would imply that the draws are sufficiently unpredictable and that each turn runs in constant time in expectation.
We present our low memory Dealer in~\Cref{sec_adaptive}, and analyze its behavior and implementation for the first~$n-m$ turns.
We treat the last~$m$ turns in~\Cref{sec_linear_perfect:sec_discussion}, where we also improve the Dealer's runtime.

The analysis of the first~$n-m$ turns (\Cref{sec_adaptive}), and that of the last~$m$ turns (\Cref{sec_linear_perfect:sec_discussion}) are of different nature and rely on different building blocks.
The separation to two sections enables us to treat each mechanism in a self-contained and focused manner.

\paragraph{Constant-Time Sampling from a Dynamic Subset:}

Given a fixed domain~$\domainSymbol$, we construct a data structure for maintaining a dynamic subset~$\subsetSymbol \subseteq \domainSymbol$. That data structure supports membership queries, addition and removal of a single element, as well as sampling of a member uniformly at random, all running in worst-case constant time, and using~$O(|\domainSymbol|)$ bits of memory (see \Cref{sec_linear_perfect:claim_maintain_and_sample_from_subset}).

We use this data structure to complete the analysis of the low memory Dealer from~\Cref{sec_adaptive} and improve its runtime.
Namely, we show that the low-memory Dealer runs in worst-case constant time per turn, and that any Guesser scores~$O(n/m + \ln m)$ points in expectation against our Dealer.
In addition, we show that there exists a Dealer with~$O(n)$ bits of memory, that runs in worst case constant time per draw, against which any Guesser is expected to score at most~$\ln n$ points.

The data structure is constructed in three layers.
Going bottom-up, we first split the domain~$\domainSymbol$ into~$|\domainSymbol|/\log|\domainSymbol|$ cells, and allocate a bit for each element in the domain, to determine its membership.
This already allows us to sample, update and query members of~$\subsetSymbol$ within a specific cell.
In the second layer, we group cells together by the number of elements from~$\subsetSymbol$ that they track, that is, by their population size.
This helps us to manage cells efficiently, as well to sample a random member of a given population size.
The third layer handles the task of sampling a population size.
Even though there are only~$\log |\domainSymbol|$ possible outcomes, we are unaware of a dynamic distribution data structure that yields the desired result. Thus, we present one of our own, which may be of independent interest.
We present this data structure in \Cref{sec_linear_perfect}.

\paragraph{Lower Bound:}
In \Cref{sec_lower_bound} we show that for any Dealer with~$m$ bits of memory, there exists a Guesser that scores~$\Omega(n /m)$ correct guesses in expectation.
We show this by establishing a connection between the size of the Dealer's memory and the Dealer's card sampling entropy. That is, we show that less memory implies more predictable draws.
Our central tool is a compression argument; namely, we show an efficient, prefix-free encoding scheme that allows us to encode and recover the course of a random game. 
Since no prefix-free encoding can surpass the entropy, we get an upper bound on the Dealer's entropy.
We present a simple Guesser called the Myopic Optimizing Guesser, which simply guesses the most probable card in each turn.
We draw a relation between the Dealer's entropy and the success probability of the Myopic Guesser, and by doing so, we prove that our low-memory Dealer is optimal.

\paragraph{Connections with Cryptography:} In \Cref{prf} we discuss the (tight) relationship between bypassing the lower bound and cryptography. The Guesser from the lower bound is not computationally efficient, and there is a good reason for that: if certain computational assumptions are true, then there is a low memory Dealer that produces a random-looking permutation where no computationally efficient Guesser can get non negligibly better than $\ln n$ correct guesses in expectation. This is when the Dealer can keep a secret. 

But in the open book case, where the Dealer has no secrets, the Myopic Guesser from~\Cref{sec_lower_bound} can still operate efficiently. We present a few research questions about the exact necessity of the existence of one-way functions to obtain our result.

\section{Preliminaries}

Throughout this paper we use $[n]$ to denote the set of integers~$\{1, \dots ,n\}$ and $[a - b]$ to denote the set of integers~$\{a, a+1, \dots, b\}$.
All logs are base $2$ unless explicitly stated otherwise, $\ln$ is the natural logarithm (base $e$).
For binary strings $s_1, s_2 \in \zo^\ast$, $|s_1|$ denotes the length in bits of $s_1$, and $s_1 \circ s_2$ denotes the concatenation of $s_1$ and $s_2$. $s[i..j]$~refers to the substring of~$s$ that spans between the~$i$th and~$j$th bits.
As a convention, we will denote random variables by bold capital letters.
We use~$\poison(x)$ to for the Poisson distribution with expected value $x$.

\subsection{Word RAM Generative Streaming Model}

\label{sec_preliminaries:sec_word_ram}

In the generative streaming model, an algorithm produces a sequence of elements, one at a time.
The span of time it takes to produce an element is referred to as a \emph{turn} or a step.
There may be an observer responding to the elements, and the algorithm may or may not take their responses into account when producing the next elements in the sequence.

We work under the assumptions of the Word RAM model in a streaming fashion.
This model is considered realistic for capturing multiple aspects of modern computers; namely, processors, operating systems, and memory. 
The main standard assumptions of this model are that an algorithm has a finite memory of $m$ bits and that this memory is made of cells (or words) where each cell consists of~$w \ge \log m$ bits. 
The algorithm can access any cell in constant time.

As for randomness, the model allows access to random bits, where the bits are produced and read ``on the fly'', so one cannot access previously read bits without storing them in memory.
We assume that we can ask for biased random bits, so that we can sample a random number~$r \sim [x]$ for any~$x \le 2^w$ in worst-case constant time. 
We discuss this in more detail in~\Cref{sec_linear_perfect:sec_discussion}.

In terms of computation, we assume that integer and bit-wise operations on~$w$ bits, occur in constant time, as is the case in modern processors.
In this paper, we assume that the processor exposes the operations \popCount, \bitSelect, and \bitRank, and that these operations run in constant time.
Given a binary string~$s$ of~$w$ bits~$s[1..w]$ the $\popCount(s)$ operation returns the number of $1$-bits in the string, that is $\popCount(s) = |\{i\in [w]: s_i = 1\}|$.
The $\bitRank(s, j)$ returns the number of $1$ bits in the substring $s[1..j]$, that is $\bitRank(s,j) = |\{i \in [j]: s_i = 1 \}|$.
The $\bitSelect(s,i)$ operation returns the index of the $i$th $1$-bit, or in other words, the least index~$j$ such that~$s[1..j]$ has exactly~$i$ 1-bits, that is $\bitSelect(s,i) = \argmin_j \bitRank(s,j) = i$.

CPUs implement \popCount since the 60s~\cite{wiki:Hamming_weight}.
And clearly, $\popCount(s)$ is just $\bitRank(s, w)$, and conversely~$\bitRank$ can be implemented by masking the higher~$w-j$ bits and calling \popCount.
If the processor does not expose any of these operations, we can preprocess a data structure that does during the initialization phase, which will use only a linear amount of bits (See Section~2 in~\cite{blandford_compact_2008}).

\subsection{Introduction to Card Guessing}
\label{sec_preliminaries:sec_intro_to_card_guessing}

Card Guessing is a game played by two players, a Guesser and a Dealer.
At the beginning of the game, the Dealer holds a deck consisting of $n$ distinct cards (labeled $1, 2, ..., n$ for simplicity).
In each turn, the Dealer picks a card from the deck, placing it face down, and the Guesser attempts to guess the identity of that card.
Once the Guesser declares its prediction, the card is revealed and discarded, and can no longer be played during this game.
The Guesser receives a point for each correct guess, and attempts to maximize the number of correct guesses throughout the game, while the Dealer attempts to be as unpredictable as possible, and reduce the amount of predicted card draws.

Say that the Dealer shuffles the deck properly at the beginning of the game, and draws cards one by one.
The capable Guesser who remembers previous draws can simply guess a card that still resides in the deck.
When the deck contains $i$ cards, any such guess is correct with probability~$1/i$.
Therefore, by linearity of expectation, the expected number of correct guesses throughout the game is~\footnote{Taken from textbook on algorithms by Kleinberg and Tardos~\cite{KT2006}, Chapter 13, Pages 721-722.} 
$$
\frac{1}{n}  + \frac{1}{n-1} + ... + \frac{1}{3} + \frac{1}{2} + 1 \approx \ln (n).
$$

In the setting considered in this paper, we assume that the Guesser has sufficient memory to remember all previously drawn cards, so $\ln n$ is the best that a Dealer can aim for.
We study the Dealer's predictability as a function of the computational resources it possesses, namely, the size of the Dealer's memory, the time spent during each turn, and its access to cryptographic primitives.

A transcript of a turn $t$ between a Guesser and a Dealer is a pair of (i) a guess $g_t \in [n]$ (made by the Guesser) and (ii) a draw of a card $d_t \in[n]$ (declared by the Dealer). The draw $d_t$ should not have appeared before. 
A transcript of a game is a sequence of guesses and draws $\{(g_\turn, d_\turn)\}_{\turn=1}^n$, such that at turn~$\turn$ the Guesser guessed the card $g_\turn$ while the Dealer drew the card~$d_\turn$.
Whenever the Guesser and the Dealer declared the same card, i.e.\ whenever $g_\turn = d_\turn$, a single point is attributed to the Guesser.
So the total \emph{score} of the game is the total number of turns at which a card was \emph{predicted correctly}.
We are interested in the \emph{expected number of correct guesses} during a game. 

At turn~$\turn$, we say that a card is {\emph available for drawing} if it has not been drawn yet.
That is, a card is available if it belongs to the set $[n] \setminus \{d_1, ... d_{\turn - 1}\}$.
A crucial requirement of the game is that the Dealer must, under all circumstances, declare an available card in each turn.
That is, the sequence~$d_1, ... d_n$ has no repetitions and thus, forms a permutation of $[n]$.

We think of the Dealer as being made up of two functions:
\begin{itemize}
    \item 
        \texttt{DRAW(turn, memory\_state, randomness):} produces the Dealer's next card draw.
        It receives the current turn, the Dealer's memory state ($m$ bits), and on-the-fly random bits.
    \item 
        \texttt{UPDATE(turn, memory\_state, last\_draw, randomness):} produces the Dealer's next memory state.
        It receives the current turn, the memory state at the beginning of the turn, on-the-fly random bits, the last card drawn (the result of \texttt{DRAW}) and \emph{optionally} the Guesser's guess.
\end{itemize}

If the Guesser can observe or deduce the Dealer's memory state, between turns, then we say that the Dealer is ``open book''.
A Dealer that does not take the Guesser's guesses into account is said to be \emph{Guesser-agnostic}.

As a warm-up, consider a Dealer with $n$ bits of memory that maintains a bitmap of available cards.
In each turn, the Dealer samples a card to pick, and if that card is available, then the Dealer draws it. 
Otherwise, the Dealer keeps sampling until an available card is found.
Tracking the cards using a bitmap takes~$n$ bits of memory.
From the coupon collector's problem, we know that the total number of sampling attempts is expected to be~$n \ln n$; thus, it is far from constant time per draw, even amortized.

\begin{construction}[\bitmapDealer] 
\label{sec_preliminaries:sec_intro_to_card_guessing:def_bitmap_dealer}
    
In each turn, the Dealer samples a card $c \in_R [n]$ uniformly at random.
    If $\card$ is available for drawing, then the Dealer declares $c$ as its draw.
    Otherwise, the dealer will sample again until an available card is found.
\end{construction}
An algorithmic description of this Dealer is presented in \Cref{sec_preliminaries:sec_intro_to_card_guessing:alg_bitmap_dealer}

\begin{algorithm}[ht]
    \caption{\bitmapDealer}
    \label{sec_preliminaries:sec_intro_to_card_guessing:alg_bitmap_dealer}
    \small
    \begin{algorithmic}
        \State
            Let $A \gets 1^n$
            \Comment{Cards availability bitmap}
        \For {turn $\turn \in [n]$}
            \State 
                Sample $c \sim [n]$ uniformly at random
            \While {$A[c] = 0$}
                \Comment{Sample until an available card is found}
                \State Sample $c \sim [n]$ uniformly at random
            \EndWhile
            \State 
                Draw card $c$
            \State
                $A[c] \gets 0$
        \EndFor
    \end{algorithmic}
\end{algorithm}

Consider also a Dealer that implements Knuth's algorithm~\cite{Knuth98}~(Section~3.4.2, Algorithm P). That Dealer explicitly stores all available cards ($n$ cards, $\log n$ bits each), samples one of them in each turn, and pops from memory.
More formally, that Dealer initiates their memory with an array~$A$ of size~$n$ (total of $n \log n$ bits), such that the $j$th cell contains~$j$, i.e.\ ${A[j] \gets j}$.
In turn $\turn$, the Dealer samples an index $i$ uniformly at random from~$[n - \turn + 1]$, draws the card stored in~$A[i]$, and assigns the card that resides in $A[n- \turn +1]$ to cell $A[i]$.
So, assuming that index $i$ was chosen at the first turn, then the memory state switches from $$
1, 2, ..., i-1, i, i+1, ..., n-1, n
$$
to $$
1, 2, ..., i-1, n, i+1, ..., n-1
.
$$

\begin{construction}[\simpleConstantTimeDealer]
    \label{sec_preliminaries:sec_intro_to_card_guessing:def_simple_constant_time}
    At the beginning of the game, the Dealer initiates an array~$A$ of~$n$ cells such that, for every~$i \in [n]: A[i] \gets i$.
    At turn $\turn$, the Dealer samples an index $i \sim [n -t + 1]$ uniformly at random and draws the card that resides at cell $A[i]$.
    Then $A[i] \gets A[n- t +1]$.
\end{construction}
An algorithmic description is provided in \Cref{sec_preliminaries:sec_intro_to_card_guessing:alg_simple_constant_time}.
And we note that this Dealer takes $O(1)$ time per draw, and~$O(n \log n)$ bits of memory.

\begin{algorithm}[ht]
    \caption{\simpleConstantTimeDealer}
    
    \label{sec_preliminaries:sec_intro_to_card_guessing:alg_simple_constant_time}
    \small
    \begin{algorithmic}
        \State
            Initiate an array of $n$ cells such that $A[i] \gets i$
        \For {turn $\turn \in [n]$}
            \State 
                Sample index $i \sim [n - \turn + 1]$ uniformly at random
            \State
                Draw card $A[i]$
            \State
                $A[i] \gets A[n - \turn + 1]$
        \EndFor
    
    \end{algorithmic}
\end{algorithm}

\subsection{Balls-into-Bins}
\label{sec_preliminaries:sec_balls_into_bins}

The Balls-into-Bins model we consider is a stochastic process where~$m$ balls are thrown into $n$ bins\footnote{In this section, we follow the common naming convention from the literature. However, in the context of our work, $m$ will be related to the number of bins, and~$n$ to the number of balls.}.
The balls are thrown in some random manner, one after another, and we are interested in the question of \emph{how many balls are there in the most (or least) loaded bin}, often with respect to the average load. 
	
In the setting of interest to this work, $m$ will be larger than $n$. This case is known as the ``heavily loaded'' case.
A probabilistic algorithm for placing balls into bins is often called an \emph{allocation process} or a \emph{sampling scheme}.	There has been much work on methods for achieving balanced allocations via simple rules of assigning the balls into the bins (e.g.\ the famous ``two-choice paradigm", where balls are thrown into two bins and are kept in the less loaded one).  See \cite{Wieder2017} for a survey on the topic.

We will actually consider two simple methods: the one choice scheme and the~\adaptiveTerm. 
Let $\loadSymbol_i(m)$ denote the load of bin $i$ after placing $m$ balls.
	\begin{algorithm}[h]
		\caption{One-choice scheme}
		\label{sec_preliminaries:sec_balls_into_bins:alg_one_choice}
        \small
		\begin{algorithmic}
			\State
				$\loadSymbol_i = 0 : \forall i \in \binsParam$
			\For {ball $v \in [\ballsParam]$}
                    \State
                    $i \sim [n]$
                    \Comment{Sample a random bin~$i$}
				\State 
					Assign ball $v$ to bin $i$
                    \State
                    $\loadSymbol_i  \pluseq 1$
			\EndFor

		\end{algorithmic}
	\end{algorithm}
	
	\begin{proposition}[Theorem 2.2 in \cite{Wieder2017}]
		\label{sec_preliminaries:sec_balls_into_bins:one_choice_gap}
After throwing $\ballsParam$ balls into $\binsParam$ bins using the one-choice scheme, with probability at least $1 - 1/\binsParam$ the maximum load is ${\frac{\ballsParam}{\binsParam}} + O\left(\sqrt{\frac{\ballsParam \ln \binsParam}{\binsParam}}\right)$ when ${\ballsParam >  \binsParam \ln \binsParam}$.
	\end{proposition}

An allocation scheme that achieves a more balanced allocation is called~\adaptiveTerm  and was introduced in~\cite{Berenbrink_Khodamoradi_Sauerwald_Stauffer_2013}.
The idea here is to make sure that no bin is too advanced.
This is achieved by placing the~$v$th ball into a bin with load at most~$v/n +1$, and by doing so, preventing bins from diverging far ahead of the average.
We will study this scheme extensively in~\Cref{sec_adaptive}.

\begin{algorithm}[ht]
\caption{\adaptiveTerm}
\label{sec_preliminaries:sec_balls_into_bins:alg_adaptive}
\small
\begin{algorithmic}
    
    \State
    $\loadSymbol_i = 0 : \forall i \in n$
    \For {ball $v \in [m]$}
    \Repeat
        \State
            $i \sim [n]$
            \Comment sample a bin uniformly at random
        \If {$\loadSymbol_i < \lceil {\frac{v}{n}}\rceil + 1$}
            \Comment If the selected bin is loaded below the threshold
            \State Place ball $v$ into bin $i$
            \State $\loadSymbol_i  \pluseq 1$
        \EndIf
        
    \Until ball is placed
    \EndFor
\end{algorithmic}
\end{algorithm}

\subsection{Data Structures}
\subsubsection{Variable bit-length arrays}
\label{sec_preliminaries:sec_vla}
An array of~$n$ cells of~$k$ bits each can easily store an ordered set of objects of size at most~$k$ while allowing fast access and update time.
If many of the stored objects require significantly less than~$k$ bits, then this scheme is wasteful. How can we dynamically adjust object lengths while maintaining efficiency?

The problem of maintaining a compact data structure for variable bit-length arrays while allowing fast look-up and update times has several solutions.
Two that match our needs are those by Blandford and Blelloch~\cite{blandford_compact_2008}, who showed a worst-case constant time compact solution for objects of known size bound, and that of PacHash~\cite{Kurpicz23PacHash}, who lifted the known size bound at the cost of having the lookup and update time constant in expectation.
Both can be implemented in our computational model.

Let~$S_w$ be the set of binary objects of size at most~$w$.
\begin{proposition}[Theorem 3.2 in~\cite{blandford_compact_2008}]
    \label{sec_preliminaries:sec_vla:thm_variable_bit_length_arrays}
    An array~$A[1,n]$ storing elements from $S_w$ with $m$ total bits
can be represented using~$O(n + m)$ bits while supporting lookups and updates in~$O(1)$ worst-case time.
\end{proposition}

\subsection{Dynamically Changing Distributions}
\label{sec_preliminaries:sec_dynamic_dist}

Consider the task of generating the outcome of a discrete random variable~$X$ that has~$n$ possible outcomes, such that~$X$ equals~$i$ with probability~$a_i$.
We can describe such a random variable as a vector~$(a_1, \dots, a_n)$ where~$\sum a_i = 1$.
The famous Alias method solves this problem; it uses at most~$O(n \log n)$ bits and runs in constant time.

What happens when the random variable~$X$ changes over time?
We now want a data structure that, in addition to initialization and generation, also supports updating the weight of some of its values, and efficiently.
There are multiple solutions to this problem, addressing different families of dynamic distributions. 

If the sum of weights is not equal to~$1$, then we say that this is a pseudo distribution, and we would like to sample the value~$j$ with probability~$a_j / \sum_{i=1}^n a_i$, where we usually consider the case where~$a_j \in \fN$.
Call a pseudo-distribution polynomially-bounded, if $\sum_{i=1}^n a_i$ is at most~$\poly(n)$ for some polynomial of degree at least~$2$.
A construction by Hagerup et. al.~\cite{hagerup1993} utilized multiple ideas by~\cite{MatiasVN93} and~\cite{rajasekaran_fast_1993}, and introduced a rejection-based algorithm to maintain and sample from polynomially bounded pseudo distributions in expected constant time.
\begin{proposition}[Theorem 2 in \cite{hagerup1993}]    \label{sec_preliminaries:sec_dynamic_dist:thm_polynomially_bounded}
    Polynomially bounded pseudo-distribution on~$[n]$ can be maintained with constant expected generation time, constant update time, $O(n \log n)$~bits of space, and~$O(n)$ initialization time.
\end{proposition}

\subsection{Information Theory and Probability}
	
\begin{definition}[Entropy]
    Let  $X$ be a discrete random variable that takes values from domain~$\cX$ with probability mass function $p(x) = \Pr[X = x]$.
    The Binary Entropy (abbreviated Entropy) of $X$, denoted $H(X)$ is
    \[
    H(X) = - \sum_{x \in \cX} p(x) \cdot \log p(x).
    \]
\end{definition}

\begin{lemma}
    \label{sec_preliminaries:lemma_entropy_heavy_element}
    Let $X$ be drawn according to the probability mass function $p$.
    If $\entropy{X} \le k$ then there exists $x \in \cX$ such that $$\Pr[X=x] \ge 2^{-k}.$$
\end{lemma}
\begin{proof}
    Assume the converse, then for every $x \in \cX: p(x) < 2^{-k}$.
    In that case, 
    \begin{flalign*}
        \entropy{X} 
        &
        = 
        \sum_{x \in \cX} p(x) \cdot \log \left(\frac{1}{p(x)}\right)
        \\
        &
        >
        \sum_{x \in \cX} p(x) \cdot \log \left(2^k\right)
        =
        k
        .
    \end{flalign*}
    In contradiction to the assumption.
\end{proof}

\begin{proposition}[Jensen's Inequality - Theorem 2.6.2 in~\cite{CoverThomas06}]
    \label{sec_preliminaries:thm_jensen_inequality}
    If $f$ is a convex function and~$X$ is a random variable, $$
\expected \left[f(X)\right] \ge f (\expected \left[X\right])
.
$$
\end{proposition}

\section{Low-Memory Dealer}

\label{sec_adaptive}

\newcommand{\allocationRV}{\mathbf{Y}}
\newcommand{\deviationRV}{\mathbf{X}}
\newcommand{\deviationInst}{x}
\newcommand{\drawableMinidecks}{D}
\newcommand{\loadRV}{\mathbf{L}}
\newcommand{\timeRV}{\mathbf{T}}
\newcommand{\averageContributionRV}{\mathbf{B}^{\stage}}
\newcommand{\averageContributionInst}{b^{\stage}}
\newcommand{\averageContributionSupp}{\cB^{\stage}}

\newcommand{\randomnessSymbol}{\cR}
\newcommand{\randomnessUntilStage}{\cR^\stage}
\newcommand{\randomnessTurn}{\cR_{\turn}}
\newcommand{\randomnessRest}{\overline{\cR_{{\turn}}^{{\stage}}}}

\newcommand{\decreaseConstant}{\kappa}
\newcommand{\lagConstant}{{c_1}}
\newcommand{\potentialParamSymbol}{\epsilon}
\newcommand{\potentialParamValue}{1/200}
\newcommand{\potentialThresholdSymbol}{\rho}
\newcommand{\potentialThresholdTerm}{(\decreaseConstant + \potentialParamSymbol)\cdot \left(1+\potentialParamSymbol\right)^\lagConstant \cdot \left({\frac{2}{\decreaseConstant}}\right)}

\renewcommand{\decreaseConstant}{\beta}
\renewcommand{\lagConstant}{\alpha}
\renewcommand{\potentialParamSymbol}{\epsilon}
\renewcommand{\potentialThresholdSymbol}{\gamma}

\newcommand{\potential}{\Phi}

\newcommand{\lagConstantLagging}{$\lagConstant$-lagging\xspace}

\newcommand{\thresholdAverage}{\lceil{\frac{\turn}{\minidecks}} \rceil}
\newcommand{\threshold}{\thresholdAverage + 1}

\newcommand{\phaseTerm}{phase\xspace}
\newcommand{\phaseAdaptive}{\adaptiveTerm \phaseTerm}
\newcommand{\phaseFinal}{Final \phaseTerm}

We present our low-memory Dealer. Recall that our goal is to have about~$m$ choices at each turn and pick one of them at random. We can think of several strategies for doing so, and the one we present is based on the ``mini-deck" approach. At the beginning of the game, the Dealer splits the deck into~$\minidecks$ mini-decks of equal size (i.e.\ $n/\minidecks$ cards in each mini-deck). 
In each turn, the Dealer chooses a mini-deck and draws the card at the top of it. So if there are~$\minidecks$ mini-decks, the Dealer will draw one of the~$\minidecks$ top cards of the different mini-decks. The advantage of picking the top card from each deck is that we simply need to recover the number of cards drawn from it, and this determines the identity of the chosen card. In terms of guessing probability, as long as all the $\minidecks$ mini-decks have cards, the probability of guessing is $1/\minidecks.$

Several complications arise: first, we need to record the number of cards drawn from each mini-deck. Naively, this takes $\log n$ bits. But even if we try to be a bit more sophisticated and record the difference between the expected number of cards drawn and the actual one, then this difference is going to be something like $\sqrt{n/\minidecks}$ (See~\Cref{sec_preliminaries:sec_balls_into_bins:one_choice_gap}). This means that it takes~$\log (n/\minidecks)$  bits per mini-deck, resulting in roughly~$m \approx \minidecks \cdot \log n$ total number of bits. Another issue is that some mini-decks may run empty before others, resulting in fewer than $m$ choices pretty early on (which means better probability for the adversary in guessing the next card). 

Instead, we aim to get a much more balanced number of cards drawn from each mini-deck, concentrated around the expected number of drawn cards. This way achieves both a compression of the bit-representation of the numbers of cards drawn and addresses the mini-deck depletion problem. The idea is to apply the \adaptiveTerm paradigm (\adaptiveTerm allocation scheme, as described in \Cref{sec_preliminaries:sec_balls_into_bins:alg_adaptive}), that is, to select at each turn a mini-deck that is below the average.
The fact that we choose from a (known) subset of all available mini-decks gives some advantage to the Guesser, but not a huge one: as we will show, its probability of guessing correctly is multiplied only by a constant.

\subsection{\adaptiveDealer}

At the beginning of the game, the Dealer splits the cards into~$\minidecks$ mini-decks of equal size and tracks their progress.
In each turn, the Dealer selects a mini-deck and draws the card that resides at the top of this mini-deck.
To keep the mini-decks balanced, the Dealer draws only from mini-decks whose progress is below a certain threshold, which is the average $+ 1$.
Let~$\loadSymbol_{i, \turn}$ be the number of cards drawn from the~$i$th mini-deck by turn~$\turn$.
At turn~$\turn$ the Dealer may draw from the~$i$th mini-deck if~$\loadSymbol_{i, \turn} < \threshold$.
The Dealer performs rejection sampling to select a mini-deck, i.e., the Dealer samples repeatedly until a suitable mini-deck is found (we will see how to do this more efficiently).
The Dealer plays this way until~$2\minidecks$ cards are left, at which point some mini-decks might run empty. We refer to the first~$n-2\minidecks$ turns as the \phaseAdaptive.
When~$2\minidecks$ turns are left, the Dealer shuffles all the remaining cards together and draws them at random, one after another.
We refer to the last~$2\minidecks$ turns as the \phaseFinal.
An algorithmic description of our Dealer is described in~\Cref{sec_adaptive:alg_adaptive_dealer}.

We now describe the Dealer's algorithm. It consists of the  the \phaseAdaptive for the first $n - 2 \minidecks$  turns, and the \phaseFinal for the last $2 \minidecks$ turns. 

\begin{algorithm}[ht]
\caption{\adaptiveDealer}
\label{sec_adaptive:alg_adaptive_dealer}
\small
\begin{algorithmic}
    \State
    $\loadSymbol_i = 0 : \forall i \in \minidecks$

    \State
    \LineComment
        \textbf{\phaseAdaptive}
    \For {turn $\turn \in [n - 2 \minidecks]$}
    \LineComment {draw a card from a suitable mini-deck}
    \Repeat
        
        \State
            $i \sim [\minidecks]$
            \Comment sample a mini-deck uniformly at random
        \If {$\loadSymbol_i < \threshold $}
            \Comment If the sampled mini-deck is less advanced than the threshold
            \State Draw the top card from the~$i$th mini-deck
            \State $\loadSymbol_i \gets \loadSymbol_i + 1$
        \EndIf
    \Until a card is drawn
    \EndFor
    \State
    \LineComment{\textbf{\phaseFinal}}
    \State
    $A \gets$ all remaining cards
    \For {turn $\turn \in \{n-2\minidecks, \dots, n\}$}
        \State
            Draw and discard a random card from~$A$
    \EndFor
    
\end{algorithmic}
\end{algorithm}

\paragraph{Implementing \Cref{sec_adaptive:alg_adaptive_dealer}:}
In terms of space, our Dealer tracks the progress of the various mini-decks by their distance from the threshold.
Let~$\deviationInst_{i, \turn}$ denote the distance of the~$i$th mini-deck from the threshold in turn~$\turn$, that is~$\deviationInst_{i, \turn} = \threshold - \loadSymbol_{i, \turn}$.
We encode the distances using Elias encoding\footnote{Any other variable-length encoding for positive integers that encodes a number~$x \in \fN$ using $O(\log(x))$ bits would work for our needs. In fact, in terms of space, we can allow ourselves to squander and encode the distances in unary, but that would come at the run time's expense.}, and store the distances in a Variable bit-Length Array, as in~\Cref{sec_preliminaries:sec_vla}.
This allows us to update and track the state of a mini-deck in constant time.
But how many mini-decks can we track this way, so that the total memory used is less than $m$?

We treat the implementation of the~\phaseFinal in~\Cref{sec_linear_perfect:sec_discussion}, since it overlaps with the implementation of the fast linear memory perfect permutation generator (note that the number of cards left is $O(d)$ and the memory is $O(d)$ as well).

\newcommand{\stage}{\tau}

Observe that the value of the threshold increases every~$\minidecks$ turns.
We refer to the span of turns during which the threshold is the same as a \emph{stage}. So there are~$n / \minidecks$ stages, and the threshold is~$\stage +1 $ during the~$\stage$th stage.
We refer to the distance between the progress of the mini-decks and the threshold by~\emph{holes}. So when it is possible to draw~$k$ more cards from a mini-deck before reaching the threshold, we say that the mini-deck has~$k$ holes\footnote{Holes is a burrowed term from the balls-into-bins literature, in which the distance between the progress of less-advanced bins and a target load (in our case, the threshold) is sometimes referred to by the term~\emph{holes}; this captures the volume absent of balls that should be filled.}.

\begin{claim}
    \label{sec_adaptive:claim_number_of_holes_smaller_than_2d}
    For any turn~$\turn$, the total number of holes is between~$\minidecks$ and~$2\minidecks$:
    $$
    \minidecks \le \sum_{i \in [\minidecks]} \deviationInst_{i,\turn} \le 2 \minidecks
    .
    $$
\end{claim}
\begin{proof}
    At the first turn of the first stage, the threshold is~$2$, and all mini-decks are full, so the Dealer can draw~$2$ cards from each mini-deck; thus, the number of holes is~$2\minidecks$.
    In each turn the Dealer draws a card, so the number of holes decreases by~$1$. 
    After~$\minidecks$ turns, the stage ends and the threshold grows by one, so the number of holes increases by~$\minidecks$.
    Overall, the number of holes is at most~$2\minidecks$ during all turns, exactly~$2\minidecks$ at the beginning of every stage, and at least~$\minidecks$.
\end{proof}
\begin{lemma}
    \label{sec_adaptive:claim_number_of_minidecks}
    The \adaptiveDealer with~$m$ bits of memory tracks~$\minidecks=\Theta(m)$ mini-decks.
\end{lemma}
\begin{proof}
    The Dealer tracks the progress of the mini-decks by their distance from the threshold and encodes the distances using an encoding function~$\encodeFunction$, such that~$\encodeFunction(x) \le O(\log x )$.
    We get that the Dealer stores~$\minidecks$ distances and that the total number of bits required to store them is:
    \begin{flalign*}
    \sum_{i \in [\minidecks]} \encodeFunction(x_{i, \turn}) 
    = 
    \sum_{i \in [\minidecks]} O(\log(x_{i, \turn}))
    = 
    \sum_{i \in [\minidecks]} O(x_{i, \turn})
    = 
    O(\minidecks)
    .
    \end{flalign*}
    By using Variable bit-Length Array, and from~\Cref{sec_preliminaries:sec_vla:thm_variable_bit_length_arrays}, we know that the Dealer can do so by using~$O(\minidecks)$ bits.
\end{proof}

\paragraph{Roadmap for the rest of the proof:}
Having settled the number of mini-decks that we can track, we would like to show that our Dealer has a sufficient number of mini-decks that it can draw cards from in every turn in the~\phaseAdaptive. This would imply that there are sufficiently many possible cards to draw in each turn, so the Dealer is as unpredictable as we want throughout the game, and that the Dealer quickly finds a suitable mini-deck.

For turn~$\turn$, we consider the vector of holes~$\deviationInst_\turn=(\deviationInst_{1, \turn}, \dots, \deviationInst_{\minidecks, \turn})$.
Many mini-decks with holes implies that there are many options for the next card, so our goal is to show that the mini-decks progress somewhat evenly throughout the game, and that the holes are spread across the mini-decks.
We apply the exponential potential function to the vector of holes and analyze the game through the lens of this function.

Set~$\potentialParamSymbol=\potentialParamValue$, and 
define the \emph{potential contribution} of the~$i$th mini-deck at the turn~$\turn$ to be~$\potential_i(\deviationInst_\turn) =  (1+\potentialParamSymbol)^{\deviationInst_{i,\turn}}$, and let the \emph{potential} of the holes vector~$\deviationInst_\turn$ at turn~$\turn$ be
$$
\potential(\deviationInst_\turn) 
= 
\sum_{i \in [\minidecks]} \potential_i(\deviationInst_\turn) 
=
\sum_{i \in [\minidecks]} (1+\potentialParamSymbol)^{\deviationInst_{i,\turn}}
.
$$

As an intuition, since there are at least~$\minidecks$ holes (\Cref{sec_adaptive:claim_number_of_holes_smaller_than_2d}), in the extreme case where all holes are in one mini-deck, then the potential is~$2^{\Omega(\minidecks)}$.
Contrary to this example, we will show that the expected value of the potential function is~$O(d)$ (in \Cref{sec_adaptive:sec_bounding}). 
We then use it to show that our Dealer is sufficiently unpredictable and that each turn takes a constant time in expectation (\Cref{sec_adaptive:sec_predictable}).

\subsection{Bounding the Potential Function}
\label{sec_adaptive:sec_bounding}

As our Dealer relies heavily on the~\adaptiveTerm allocation scheme~(\Cref{sec_preliminaries:sec_balls_into_bins:alg_adaptive}), we will walk the path paved by~Berenbrink et al.\cite{Berenbrink_Khodamoradi_Sauerwald_Stauffer_2013} to show our results.

Our main machinery is the exponential potential function that we apply to the vector of holes.
This potential function is monotone in the sense that mini-decks that lag behind, that is, have more holes, contribute more to the value of the potential function.
We start by showing that there is some constant~$\lagConstant$ such that, if a mini-deck has at least~$\lagConstant$ holes at the end of a stage, then we expect to draw strictly more than one card from it during the following stage, thus, it progresses faster and has fewer holes than at the beginning of the stage.
As a result, we expect its potential contribution to decrease.
We then show that the potential decrease caused by lagging mini-decks suffices to bound the expected potential value of the vector of holes throughout the entire game.

As a convention, we will use Greek letters for constants, Capital letters for sets, {\bf bold} ones for random variables, and lower letters for specific values of the corresponding random variables.
We will often change focus between turns and stages (recall that stages are made of~$\minidecks$ consecutive turns).
Therefore, for clarity, we will use subscripted~${}_{\turn}$ to refer to the state of objects during turn~$\turn$ and superscripted~${}^{\stage}$ when referring to stages.

Let $\loadRV_i^\stage$ be the random variable measuring the number of cards drawn from the~$i$th mini-deck by the end of stage~$\stage$, and consider the ``global'' random variable vector that represents the progress of all mini-decks~$\loadRV^\stage =(\loadRV_1^\stage, \dots, \loadRV_\minidecks^\stage)$.
Consider the equivalent random variable~$\deviationRV^{\stage}$ that describes the number of holes in the~$i$th mini-deck by the end of stage~$\stage$ and is defined as~$\deviationRV_i^{\stage} = \stage + 2 - \loadRV_i^\stage$, and the random vector of deviation at the end of the stage~$\deviationRV^{\stage}=(\deviationRV_1^{\stage}, \dots, \deviationRV_\minidecks^{\stage})$.
While we are at it, let~$\loadRV_{i, \turn}, \loadRV_{\turn}, \deviationRV_{i, \turn}, \deviationRV_{\turn}$ be the corresponding random variables at a specific turn~$\turn$.
Let $\allocationRV_i^{\stage}$ be the random variable measuring the number of cards drawn from the~$i$th mini-deck at stage~$\stage$.

We start by showing that mini-decks that are too far from the threshold at some stage are likely to get closer to the threshold in the following stage; moreover, we expect to draw slightly more than~$1$ card from such mini-decks.
Call a mini-deck \emph{\lagConstantLagging} at the end of stage~$\stage$ if it has more than~$\lagConstant$ holes, for some constant~$\lagConstant$.

\begin{proposition}[Lemma 3.3 in~\cite{Berenbrink_Khodamoradi_Sauerwald_Stauffer_2013}]
\label{sec_adaptive:claim_lagging_mini_deck_yields_more_than_one_draw}
    There exists a constant~$\lagConstant$, such that for any stage~$\stage$ and any~$0 \le k \le \lagConstant$,
    if mini-deck~$i$ is \lagConstantLagging at the end of stage~$\stage$, then the probability that the \adaptiveDealer draws at least~$k$ cards from the~$i$th mini-deck at stage~$\stage+1$ is at least
    $$
    \Pr[\allocationRV_i^{\stage+1} \ge k] 
    \ge 
    \Pr[\poison(199/198) \ge k] - 2 \cdot 10^{-10}
    .
    $$
\end{proposition}

We use the exponential potential function to analyze the Dealer's progress throughout the game. 
Consider some deviation vector~$\deviationInst^{\stage} = (\deviationInst^{\stage}_1, \dots, \deviationInst^{\stage}_\minidecks)$ that captures the mini-decks holes by the end of stage~$\stage$.
And recall that the further a mini-deck is from the threshold, the larger is its potential contribution. 
At the same time, observe that the potential contribution cannot grow much between stages.
\begin{claim}
    \label{sec_adaptive:claim_potential_increase_general_upper_bound}
    For any mini-deck~$i \in [\minidecks]$ and any stage~$\stage$, the potential contribution of the~$i$th mini-deck grows by a factor of at most~$(1+\potentialParamSymbol)$ by the end of stage~$\stage+1$:
    $$
    \potential_i(\deviationInst^{\stage+1})
    \le 
    (1+\potentialParamSymbol) \cdot \potential_i(\deviationInst^{\stage})
    .
    $$
\end{claim}
\begin{proof}
    It could be the case that the Dealer doesn't sample the $i$th mini-deck during a stage, resulting in the mini-deck~$i$ getting farther from the threshold by~$1$, and that's the worst that can happen.
\end{proof}

The following \namecref{sec_adaptive:claim_potential_contribution_decrease_of_lagging_mini_decks} shows that there exists a constant~$0 <\decreaseConstant < 1$ that depends on~$\lagConstant$, such that the potential contribution of $\lagConstant$-lagging mini-decks decreases in expectation by a factor of~$(1-\decreaseConstant)$.
\begin{proposition}[Lemma 3.4 in~\cite{Berenbrink_Khodamoradi_Sauerwald_Stauffer_2013}]
    \label{sec_adaptive:claim_potential_contribution_decrease_of_lagging_mini_decks}
    If mini-deck~$i$ is \lagConstantLagging at the end of stage~$\stage$, then its potential contribution is expected to decrease at least by a factor of~$(1-\decreaseConstant)$ on stage~$\stage+1$
    $$
    \fE\left[\potential_i(\deviationRV^{\stage+1}) | \deviationRV^{\stage} = \deviationInst^\stage \right] 
    \le
    (1 - \decreaseConstant) \cdot \potential_i(\deviationInst^{\stage})
    .
    $$
\end{proposition}

Let~$\potentialThresholdSymbol = \potentialThresholdTerm$ be a potential threshold constant.
We use~\Cref{sec_adaptive:claim_potential_contribution_decrease_of_lagging_mini_decks} to show that if the potential is too big, then the potential contribution decrease of lagging mini-decks is expected to reduce the potential of the entire system by the end of the following stage.

\begin{lemma}
    \label{sec_adaptive:claim_potential_decreases_if_too_large}
    For any stage~$\stage$, and any holes vector~$\deviationInst^\stage$ such that~$\deviationRV^\stage = \deviationInst^\stage$, if the potential~$\potential(\deviationInst^\stage)$ is larger than~$\potentialThresholdSymbol \cdot \minidecks$, then the potential of~$\deviationRV^{\stage+1}$ is expected to decrease by a factor of~$\left(1 - {\frac{\decreaseConstant}{2}}\right)$ by the end of stage~$\stage+1$,
    $$
    \fE[\potential(\deviationRV^{\stage+1}) | \deviationRV^{\stage} = \deviationInst^{\stage}]
    \le
    \left(1 - {\frac{\decreaseConstant}{2}}\right) \cdot \potential(\deviationInst^{\stage})
    .
    $$
    
\end{lemma}
\begin{proof}
    \label{sec_adaptive:prf_potential_decreases_if_too_large}
    \newcommand{\drawableNotLag}{\drawableMinidecks_{\le\lagConstant}^{\stage}}
    
    Consider the subset of mini-decks~$\drawableNotLag \subseteq [\minidecks]$ such that, at the end of stage~$\stage$, are not \lagConstantLagging, i.e.~mini-deck~$i\in \drawableNotLag$ if it has at most~$\lagConstant$ holes.
    We get that:
    \begin{flalign*}
        \fE\left[\potential(\deviationRV^{\stage+1}) | \deviationRV^{\stage} = \deviationInst^\stage \right] 
        &
        =
        \sum_{i \in [\minidecks]}
        \fE\left[\potential_i(\deviationRV^{\stage+1}) | \deviationRV^{\stage} = \deviationInst^\stage \right] 
        \nonumber
        \\
        &
        =
        \sum_{i \notin \drawableNotLag}
        \fE\left[\potential_i(\deviationRV^{\stage+1}) | \deviationRV^{\stage} = \deviationInst^\stage \right] 
        +
        \sum_{i \in \drawableNotLag}
        \fE\left[\potential_i(\deviationRV^{\stage+1}) | \deviationRV^{\stage} = \deviationInst^\stage \right]  = (\ast)
        .
    \end{flalign*}
    
    For the second summand, we use the general potential growth upper bound from~\Cref{sec_adaptive:claim_potential_increase_general_upper_bound} to get that
    $
        \sum_{i \in \drawableNotLag}
        \fE\left[\potential_i(\deviationRV^{\stage+1}) | \deviationRV^{\stage} = \deviationInst^\stage \right] 
        \le
        (1+\potentialParamSymbol)
        \cdot
        \sum_{i \in \drawableNotLag} \potential_i(\deviationInst^\stage)
    $.
    And for the first summand, we apply~\Cref{sec_adaptive:claim_potential_contribution_decrease_of_lagging_mini_decks} to conclude that
    \begin{flalign*}
        \sum_{i \notin \drawableNotLag}
        \fE\left[\potential(\deviationRV^{\stage+1}) | \deviationRV^{\stage} = \deviationInst^\stage \right] 
        &
        \le
        (1-\decreaseConstant) \cdot
        \sum_{i \notin \drawableNotLag}
        \potential_i(\deviationInst^\stage)
        \\
        &
        =
        (1-\decreaseConstant) \cdot
        \left(
        \sum_{i \in [\minidecks]}
        \potential_i(\deviationInst^\stage)
        - 
        \sum_{i \in \drawableNotLag}
        \potential_i(\deviationInst^\stage)
        \right)
        \nonumber
        \\
        &
        =
        (1-\decreaseConstant) \cdot \potential(\deviationInst^\stage)
        +
        (\decreaseConstant-1) \cdot \sum_{i \in \drawableNotLag}
        \potential_i(\deviationInst^\stage)
        .
    \end{flalign*}

    Plugging these together into~$(\ast)$, we get that
    \begin{flalign}
         \fE\left[\potential(\deviationRV^{\stage+1}) | \deviationRV^{\stage} = \deviationInst^\stage \right] 
         &
         \le
        \left[(1-\decreaseConstant) \cdot \potential(\deviationInst^\stage) 
        +
        (\decreaseConstant-1) \cdot \sum_{i \in \drawableNotLag}
        \potential_i(\deviationInst^\stage) \right]
        +
        \left[
        (1+\potentialParamSymbol)
        \cdot
        \sum_{i \in \drawableNotLag} \potential_i(\deviationInst^\stage)
        \right]
        \nonumber
        \\
        &
        =
        (1-\decreaseConstant) \cdot \potential(\deviationInst^\stage)
        +
        (\decreaseConstant+\potentialParamSymbol) \cdot \sum_{i \in \drawableNotLag}
        \potential_i(\deviationInst^\stage)
        \nonumber
        \\
        &
        \le
        (1-\decreaseConstant) \cdot \potential(\deviationInst^\stage)
        +
        (\decreaseConstant+\potentialParamSymbol) \cdot \sum_{i \in \drawableNotLag}
        (1+\epsilon)^\lagConstant
        \label[ineq]{sec_adaptive:prf_potential_decreases_if_too_large:set_of_bounded_holes}
        \\
        &
        \le
        (1-\decreaseConstant) \cdot \potential(\deviationInst^\stage)
        +
        (\decreaseConstant+\potentialParamSymbol) \cdot \minidecks \cdot
        (1+\epsilon)^\lagConstant
        \label[ineq]{sec_adaptive:prf_potential_decreases_if_too_large:subset_is_at_most_set}
        \\
        &
        \le
        (1-\decreaseConstant) \cdot \potential(\deviationInst^\stage)
        +
        {\frac{\decreaseConstant}{2}} \cdot \potential (\deviationInst^\stage)
        \label[ineq]{sec_adaptive:prf_potential_decreases_if_too_large:by_assumption_on_potential_threshold}
        \\
        &
        =
        \left(1 - {\frac{\decreaseConstant}{2}}\right) \cdot \potential (\deviationInst^\stage)
        \nonumber
         .
    \end{flalign}
    Where
    \Cref{sec_adaptive:prf_potential_decreases_if_too_large:set_of_bounded_holes} is true because for every~$i \in \drawableNotLag$ it holds that~$\potential_i(\deviationInst^\stage) \le (1+\potentialParamSymbol)^\lagConstant$,
    \Cref{sec_adaptive:prf_potential_decreases_if_too_large:subset_is_at_most_set} is true because~$|\drawableNotLag| \le \minidecks$,
    and
    \Cref{sec_adaptive:prf_potential_decreases_if_too_large:by_assumption_on_potential_threshold} follows from our assumption that the potential is larger than
    $\potential(\deviationInst^\stage) \ge \potentialThresholdTerm \cdot \minidecks$.
\end{proof}

\begin{corollary}
    \label{sec_adaptive:claim_expected_potential_upper_bound}
    For every stage~$\stage$, the expected potential is at most
    $$
    \fE\left[\potential\left(\deviationRV^\stage\right)\right]
    \le
    \minidecks \cdot \left( (1 + \potentialParamSymbol)^2 \cdot \potentialThresholdSymbol \cdot  \left ({\frac{2}{ \decreaseConstant}} \right) \right)
    =
    O(\minidecks)
    .
    $$
\end{corollary}

\begin{proof}
    \label{sec_adaptive:prf_expected_potential_upper_bound}
    We first observe that at the beginning of the very first stage, there are two holes in each mini-deck, and therefore
    $\potential(\deviationInst^0) = \minidecks \cdot (1+\potentialParamSymbol)^2 = O(\minidecks)$.
    Assume by induction that the statement is true for all stages up until~$\stage-1$, and we will show that it is also true for stage~$\stage$.

    First, consider the case where~$\fE[\potential(\deviationRV^{\stage-1})] \le (1+\potentialParamSymbol)\cdot \potentialThresholdSymbol \cdot \minidecks \cdot \left ({\frac{2}{\decreaseConstant}} \right)$:
    \begin{flalign}
         \fE\left[\potential(\deviationRV^{\stage})\right]
         &
         \le
         (1 + \potentialParamSymbol) \cdot  \fE\left[\potential(\deviationRV^{\stage-1})\right]
         \label[ineq]{sec_adaptive:prf_expected_potential_upper_bound:by_claim_potential_increase_general_upper_bound}
         \\
         &
         \le
         (1+\potentialParamSymbol)^2 \cdot \potentialThresholdSymbol \cdot \minidecks \cdot \left (\frac{2}{\decreaseConstant} \right)
         \label[ineq]{sec_adaptive:prf_expected_potential_upper_bound:by_assumption}
         .
    \end{flalign}
    Where
    \Cref{sec_adaptive:prf_expected_potential_upper_bound:by_claim_potential_increase_general_upper_bound} follows from~\Cref{sec_adaptive:claim_potential_increase_general_upper_bound} and monotonicity of expectation,
    and 
    \Cref{sec_adaptive:prf_expected_potential_upper_bound:by_assumption} is true by the assumption that~$\fE\left[\potential(\deviationRV^{\stage-1})\right] \le (1+\potentialParamSymbol)\cdot \potentialThresholdSymbol \cdot \minidecks \cdot \left (\frac{2}{\decreaseConstant} \right)$.

    Otherwise, consider the case where~$\fE\left[\potential(\deviationRV^{\stage-1})\right] \ge (1+\potentialParamSymbol)\cdot \potentialThresholdSymbol \cdot \minidecks \cdot \left (\frac{2}{\decreaseConstant} \right)$.
    There are two options now, either the potential at the previous stage is above~$\potentialThresholdSymbol \cdot \minidecks$ or below it. From the law of total expectation, we get that
    \begin{flalign}
        \fE\left[\potential(\deviationRV^{\stage})\right]
        =
        &
        \Pr\left[\potential(\deviationRV^{\stage-1}) \le \potentialThresholdSymbol \cdot \minidecks \right] 
        \cdot 
        \underbrace{\fE\left[\potential(\deviationRV^{\stage}) | \potential(\deviationRV^{\stage-1}) \le \potentialThresholdSymbol \cdot \minidecks \right]}_{(\ast)}
        \nonumber
        \\
        &
        +
        \Pr\left[\potential(\deviationRV^{\stage-1}) > \potentialThresholdSymbol \cdot \minidecks\right] 
        \cdot 
        \underbrace{\fE\left[\potential(\deviationRV^{\stage}) | \potential(\deviationRV^{\stage-1}) > \potentialThresholdSymbol \cdot \minidecks \right]}_{(\ast\ast)}
        .
        \label[equa]{sec_adaptive:prf_expected_potential_upper_bound:separation_to_cases}
    \end{flalign}
    We handle the two terms separately. 
    \begin{flalign}
        (\ast)
        &
        \le
        (1 + \potentialParamSymbol) \cdot \potentialThresholdSymbol \cdot \minidecks
        \label[ineq]{sec_adaptive:prf_expected_potential_upper_bound:by_claim_potential_increase_general_upper_bound_a}
        \\
        &
        \le
        \left(\frac{\decreaseConstant}{2}\right) \cdot \fE\left[\potential(\deviationRV^{\stage-1})\right]
        \label[ineq]{sec_adaptive:prf_expected_potential_upper_bound:case_a}
    .
    \end{flalign}
    Where
    \Cref{sec_adaptive:prf_expected_potential_upper_bound:by_claim_potential_increase_general_upper_bound_a} follows from the general potential increase upper bound described in~\Cref{sec_adaptive:claim_potential_increase_general_upper_bound},
    and~\Cref{sec_adaptive:prf_expected_potential_upper_bound:case_a}
    follows from the assumption that~$\fE\left[\potential(\deviationRV^{\stage-1})\right] \ge (1+\potentialParamSymbol)\cdot \potentialThresholdSymbol \cdot \minidecks \cdot \left (\frac{2}{\decreaseConstant} \right)$.

    As for the second term, we apply~\Cref{sec_adaptive:claim_potential_decreases_if_too_large} to conclude that
    \begin{flalign}
        (\ast\ast) 
        &
        \le
        \left(1 - \frac{\decreaseConstant}{2}\right) \cdot
        \fE\left[\potential(\deviationRV^{\stage-1})\right]
        \label[ineq]{sec_adaptive:prf_expected_potential_upper_bound:case_b}
        .
    \end{flalign}
    
    We observe that the right-hand side of~\Cref{sec_adaptive:prf_expected_potential_upper_bound:separation_to_cases} is a convex combination of~$(\ast)$ and~$(\ast\ast)$, thus it is clearly smaller than their sum.
    Placing \Cref{sec_adaptive:prf_expected_potential_upper_bound:case_a} and~\Cref{sec_adaptive:prf_expected_potential_upper_bound:case_b} into~\Cref{sec_adaptive:prf_expected_potential_upper_bound:separation_to_cases} we get that
    \begin{flalign*}
        \fE\left[\potential(\deviationRV^{\stage})\right]
        &
        <
        \left(\frac{\decreaseConstant}{2}\right) \cdot \fE\left[\potential(\deviationRV^{\stage-1})\right]
        +
        \left(1 - \frac{\decreaseConstant}{2}\right) \cdot
        \fE\left[\potential(\deviationRV^{\stage-1})\right]
        \\
        &
        =
        \fE\left[\potential(\deviationRV^{\stage-1})\right]
        .
    \end{flalign*}
    Which, by our inductive assumption, satisfies the desired bound.
\end{proof}

\subsection{Predictability and Run Time}
\label{sec_adaptive:sec_predictable}
We now show that our Dealer is sufficiently unpredictable and that each turn takes a constant time in expectation.
To do so, we focus on the set of mini-decks whose potential contribution is relatively small (compared to the average contribution), and we show that they form a constant fraction of {\em all} mini-decks, in every turn.
This implies that (i) the rejection-sampling process is expected to take only a constant time. Since all other parts run in worst-case constant time, we get that the Dealer runs in constant time in expectation. And (ii) that the Dealer is expected to have many mini-decks to draw from, thus, the Guesser's expected benefit is small. 
We show that this is the case for all turns, and so, by linearity of expectation, we conclude the expected score of our Dealer.

For a stage~$\stage$, and a vector of holes $\deviationRV^\stage$ at the end of stage~$\stage$, we introduce a new random variable~$\averageContributionRV = \max\left\{\frac{\potential(\deviationRV^{\stage})}{\minidecks}, \frac{16}{ \potentialParamSymbol^2}\right\}$, which can be roughly thought of as the average potential contribution of the various mini-decks.

\begin{claim}
    \label{sec_adaptive:claim_expected_average_contribution_is_constant}
    For every stage~$\stage$, $\fE[\averageContributionRV] = O(1).$
\end{claim}
\begin{proof}
    Consider the constants~$c_1 = (1 + \potentialParamSymbol)^2 \cdot \potentialThresholdSymbol \cdot\left ({2 \over \decreaseConstant}\right)$ and~$c_2 = {16 / \potentialParamSymbol^2}$. 
From~\Cref{sec_adaptive:claim_expected_potential_upper_bound} we know that
$$\fE\left[\potential\left(\deviationRV^\stage\right)\right] \le c_1 \cdot \minidecks,$$ so we can conclude that~$\fE\left[{\potential\left(\deviationRV^\stage\right) \over \minidecks}\right] \le c_1.$
    Considering that~$\averageContributionRV$ is either~${\potential\left(\deviationRV^\stage\right) \over \minidecks}$ or~$c_2$, and using the law of total expectation, we get that~$\fE[\averageContributionRV]$ is a convex combination of two constants, thus upper bounded by a constant.
\end{proof}

\newcommand{\drawableKHoles}{\drawableMinidecks_{k}^{\stage}}

Observe that for any turn there is a finite set of possible values for the holes vector, and thus, for the possible potential values.
Let~$\averageContributionSupp$ be the set of possible values for~$\averageContributionRV$.
Let $\drawableKHoles$ be the set of mini-decks that have exactly~$k$ holes by the end of the~$\stage$th stage; that is, for every mini-deck~$i \in \drawableKHoles$ it holds that~$\potential_i(\deviationInst^{\stage}) = (1+\epsilon)^k$.

\begin{claim}
    \label{sec_adaptive:claim_number_of_bins_with_k_holes_ub}
    For every stage~$\stage$, every~$\averageContributionInst \in \averageContributionSupp$, and every~$k \in [\stage+1]$,
    if the average potential contribution at the end of stage~$\stage$ is~$\averageContributionInst$, 
    then, the number of mini-decks with~$k$ holes is at most
    $$
    |\drawableKHoles| 
    \le 
    \averageContributionInst \cdot \minidecks \cdot 2^{-k\potentialParamSymbol}
    .
    $$
\end{claim}
\begin{proof}
    Let~$\deviationInst^\stage$ be any hole vector at the end of stage~$\stage$ such that the average potential contribution is~$\averageContributionInst$, then
    \begin{flalign*}
        \averageContributionInst \cdot \minidecks
        &
        \ge
        \potential\left(\deviationInst^{\stage}\right)
        \defeq
        \sum_{i \in [\minidecks]} \potential_i(\deviationInst^{\stage})
        \nonumber
        \ge
        \sum_{i \in \drawableKHoles} (1+\potentialParamSymbol)^{k}
        =
        \left|\drawableKHoles\right| \cdot (1+\potentialParamSymbol)^{k}
        \nonumber
        .
    \end{flalign*}
    The first inequality is true by the definition of~$\averageContributionRV$,
    and the second inequality is true by the definition of~$\drawableKHoles$ and since $\drawableKHoles \subseteq [\minidecks]$.

    We get that 
    \begin{flalign*}
        \left|\drawableKHoles\right| 
        \le
        \averageContributionInst \cdot \minidecks \cdot (1+\potentialParamSymbol)^{-k}
        \le
        \averageContributionInst \cdot \minidecks \cdot 2^{- k \potentialParamSymbol}
    \end{flalign*}
    where the last inequality is true because $(1+\potentialParamSymbol)^{-1} \le 2^{-\potentialParamSymbol}$ for $0 \le \potentialParamSymbol \le 1$.
\end{proof}

Observe that a mini-deck can have at most~$\stage+2$ holes during stage~$\stage$.
We use this fact to conclude an upper bound on the number of holes in mini-decks that are relatively far behind.
\begin{claim}
    \label{sec_adaptive:claim_number_of_holes_at_lagging_decks}
    For every stage~$\stage$ and every~$\averageContributionInst \in \averageContributionSupp$,
    if the average potential contribution at the end of stage~$\stage$ is~$\averageContributionInst$, then the total number of holes in mini-decks having at least~$4 \ln (\averageContributionInst) \over \potentialParamSymbol$ holes is at most 
    $$
    \sum_{k={4 \ln (\averageContributionInst)/ \potentialParamSymbol}}^{\stage+2}
    |\drawableKHoles| \cdot k
    \le
    \minidecks \cdot {8 \over \potentialParamSymbol^2} \cdot {1 \over \averageContributionInst}
    .
    $$
\end{claim}
\begin{proof}
    
    \begin{flalign}
        \sum_{k=4 \ln(\averageContributionInst) / \potentialParamSymbol}^{\stage+2} 
        |\drawableKHoles| \cdot k
        &
        \le
        \sum_{k=4 \ln(\averageContributionInst) / \potentialParamSymbol}^{\stage+2} 
        \averageContributionInst \cdot \minidecks \cdot 2^{-k\potentialParamSymbol} \cdot k
        \label[ineq]{sec_adaptive:prf_number_of_holes_at_lagging_decks:from_claim_number_of_bins_with_k_holes_ub}
        \\
        &
        =
        {\averageContributionInst \cdot \minidecks \over \potentialParamSymbol} \cdot
        \sum_{k=4 \ln(\averageContributionInst) / \potentialParamSymbol}^{\stage+2} 
        2^{-k\potentialParamSymbol} 
        \cdot k \cdot \potentialParamSymbol
        \nonumber
        \\
        &
        =
        {\averageContributionInst \cdot \minidecks \over \potentialParamSymbol} \cdot
        \sum_{k=4 \ln(\averageContributionInst) / \potentialParamSymbol}^{\stage+2} 
        e^{-k \potentialParamSymbol /2}
        \cdot
        \left(
        4 \over e
        \right)^{{-k \potentialParamSymbol /2}}
        \cdot k \cdot \potentialParamSymbol
        \nonumber
        \\
        &
        <
        {\averageContributionInst \cdot \minidecks \over \potentialParamSymbol} \cdot
        \sum_{k=4 \ln(\averageContributionInst) / \potentialParamSymbol}^{\stage+2} 
        e^{-k \potentialParamSymbol /2}
        \cdot 2
        \label[ineq]{sec_adaptive:prf_number_of_holes_at_lagging_decks:by_upper_bound_on_term}
        \\
        &
        <
        2 \cdot {\averageContributionInst \cdot \minidecks \over \potentialParamSymbol} \cdot
        {
        e^{-2 \cdot \ln(\averageContributionInst)}
        \over
        1-e^{\potentialParamSymbol/2}
        }
        \label[ineq]{sec_adaptive:prf_number_of_holes_at_lagging_decks:by_sum_of_geometric_series}
        \\
        &
        <
        2 \cdot {\averageContributionInst \cdot \minidecks \over \potentialParamSymbol} \cdot
        {
        (\averageContributionInst)^{-2}
        \over
        \potentialParamSymbol/4
        }
        \label[ineq]{sec_adaptive:prf_number_of_holes_at_lagging_decks:just_ub}
        =
        \minidecks \cdot {8 \over \potentialParamSymbol^2} \cdot {1 \over \averageContributionInst}
        .
    \end{flalign}
    Where \Cref{sec_adaptive:prf_number_of_holes_at_lagging_decks:from_claim_number_of_bins_with_k_holes_ub} follows from~\Cref{sec_adaptive:claim_number_of_bins_with_k_holes_ub},
    \Cref{sec_adaptive:prf_number_of_holes_at_lagging_decks:by_upper_bound_on_term} is true since~$(4/e)^{x/2} \cdot x < 2$ for every~$x$.
    We deduce \Cref{sec_adaptive:prf_number_of_holes_at_lagging_decks:by_sum_of_geometric_series} from the closed form of the sum of the corresponding infinite geometric series, 
    and
    \Cref{sec_adaptive:prf_number_of_holes_at_lagging_decks:just_ub} is true because~$x/2 \le 1 - e^{-x}$ for $0 \le x \le 1$.
\end{proof}

\newcommand{\drawableAtLeast}{\drawableMinidecks'_{\turn}}
\newcommand{\drawableAtLeastOtherTurn}{\drawableMinidecks'_{\turn'}}
For a turn~$\turn$ in stage~$\stage$, such that~$\averageContributionRV=\averageContributionInst$ for some value~$\averageContributionInst$, let~$\drawableAtLeast$ be the set of mini-decks that have at least~$1$ hole and at most~$4 \ln(\averageContributionInst)/ \potentialParamSymbol$ holes at turn~$\turn$.

\begin{claim}
\label{sec_adaptive:claim_number_of_drawable_minidecks_lb}
For every stage~$\stage$, and every value~$\averageContributionInst \in \averageContributionSupp$, 
if the average potential contribution at the end of stage~$\stage$ is~$\averageContributionInst$, 
then for every turn~$\turn$ in stage~$\stage+1$, the set of mini-decks with at least~$1$ hole and at most~$4 \ln(\averageContributionInst)/ \potentialParamSymbol$ holes at turn~$\turn$ is at least
$$
    |\drawableAtLeast| = {\potentialParamSymbol \over 8} \cdot {\minidecks \over \ln(\averageContributionInst)} = \Omega\left(\minidecks \over \ln(\averageContributionInst)\right)
$$
\end{claim}
\begin{proof}
    Recall (as appears in the proof of ~\Cref{sec_adaptive:claim_number_of_holes_smaller_than_2d}) that the number of holes is exactly~$2\minidecks$ at the beginning of each stage.
    Let~$\turn'$ be the first turn at stage~$\stage+1$.
    From~\Cref{sec_adaptive:claim_number_of_holes_at_lagging_decks} we conclude that the number of holes over mini-decks in~$\drawableAtLeastOtherTurn$ is at least~$2 \minidecks - \minidecks \cdot {8 \over \potentialParamSymbol^2} \cdot {1 \over \averageContributionInst} 
    \ge 
    2 \minidecks - \minidecks \cdot {8 \over \potentialParamSymbol^2} \cdot {\potentialParamSymbol^2 \over 16} = {3 \over 2} \minidecks$
    .
    It follows that even if the Dealer draws only from mini-decks in~$\drawableAtLeast$, then after~$\minidecks-1$ draws, there are at least~$\minidecks/2$ holes in mini-decks from $\drawableAtLeast$. 
    Thus, this is true for every turn~$\turn$ in the stage.
    
    Because $\drawableAtLeast$ contains mini-decks with at most~$4 \ln(\averageContributionInst)/\potentialParamSymbol$, we get that 
    $$
    |\drawableAtLeast| 
    \ge 
    {\minidecks \over 2} \cdot {\potentialParamSymbol \over 4 \ln(\averageContributionInst)}
    =
    \Omega\left(\minidecks \over \ln(\averageContributionInst)\right)
    .
    $$
\end{proof}

The following lemmata combine the lower bound on the size of relatively advanced mini-decks~$\drawableAtLeast$~(\Cref{sec_adaptive:claim_number_of_drawable_minidecks_lb}) and the upper bound on the expectation of the potential~(\Cref{sec_adaptive:claim_expected_average_contribution_is_constant}), to conclude the Dealer's run time and predictability.
Their proofs share a similar structure and analysis. 

Recall that the Dealer samples mini-decks until a drawable one is found.
Let~$\timeRV_\turn$ be the random variable measuring the time it takes the Dealer to draw a card at turn~$\turn$.

\begin{lemma}
\label{sec_adaptive:claim_expected_runtime_at_turn}
For every turn~$\turn$, the Dealer's expected run time at turn~$\turn$ is~$\fE[\timeRV_\turn] = O(1).$
\end{lemma}
\begin{proof}
    \label{sec_adaptive:proof_expected_runtime_at_turn}
    Let~$\stage$ be the stage prior to turn~$\turn$.
    Let~$\drawableMinidecks_\turn \subseteq [\minidecks]$ be the set of all drawable mini-decks at turn~$\turn$.
    Clearly, $\drawableAtLeast \subseteq \drawableMinidecks_\turn$, and therefore, the expected number of trials to find a mini-deck from~$\drawableMinidecks_\turn$ is upper bounded by the expected number of trials to find a mini-deck from~$\drawableAtLeast$, the later is a geometric random variable with probability~$|\drawableAtLeast|/ \minidecks$.
    Using the lower bound on~$\drawableAtLeast$ from \Cref{sec_adaptive:claim_number_of_drawable_minidecks_lb} we get that if~$\averageContributionRV=\averageContributionInst$ then expected time for success is~$O(\ln(\averageContributionInst))$, where the expectation is taken solely over the randomness of the dealer at turn~$\turn$.

    Let~$\averageContributionSupp$ be the set of possible values for~$\averageContributionRV$.
    We get that
    \begin{flalign}
        \expectation\left[\timeRV_\turn\right]
        &
        =
        \sum_{\averageContributionInst \in \averageContributionSupp}
        \expectation\left[\timeRV_\turn | \averageContributionRV = \averageContributionInst\right] \cdot \Pr\left[\averageContributionRV = \averageContributionInst\right]
        \label[equa]{sec_adaptive:proof_expected_runtime_at_turn:from_total_expectation}
        \\
        &
        \le
        \sum_{\averageContributionInst \in \averageContributionSupp}
        O(\ln(\averageContributionInst))
        \cdot
        \Pr\left[\averageContributionRV = \averageContributionInst\right]
        \label[ineq]{sec_adaptive:proof_expected_runtime_at_turn:by_the_discussion_above}
        \\
        &
        =
        \expectation \left[O(\ln(\averageContributionRV)) \right]
        \label[equa]{sec_adaptive:proof_expected_runtime_at_turn:by_definition_of_expectation}
        \\
        &
        \le
        O\left(\ln\left(\expectation \left[\averageContributionRV) \right] \right)\right)
        \label[ineq]{sec_adaptive:proof_expected_runtime_at_turn:from_jensen_inequality_and_linearity_of_expectation}
        \\
        &
        \le
        O(1)
        \label[ineq]{sec_adaptive:proof_expected_runtime_at_turn:from_claim_expected_average_contribution_is_constant}
        .
    \end{flalign}
    Where
    \Cref{sec_adaptive:proof_expected_runtime_at_turn:from_total_expectation} follows from the law of total expectation,
    \Cref{sec_adaptive:proof_expected_runtime_at_turn:by_the_discussion_above} is true by the discussion above,
    \Cref{sec_adaptive:proof_expected_runtime_at_turn:by_definition_of_expectation} is true by the definition of expectation,
    \Cref{sec_adaptive:proof_expected_runtime_at_turn:from_jensen_inequality_and_linearity_of_expectation} follows from the linearity of expectation and from Jensen inequality~\Cref{sec_preliminaries:thm_jensen_inequality},
    and~\Cref{sec_adaptive:proof_expected_runtime_at_turn:from_claim_expected_average_contribution_is_constant} follows from~\Cref{sec_adaptive:claim_expected_average_contribution_is_constant}.
    
\end{proof}

Observe that the Guesser's expected benefit is bounded from above by the number of drawable mini-decks.
Let the indicator random variable~$\correctTurnRV$ be the event that the Guesser predicted the Dealer's draw correctly at turn~$\turn$, and thus scored a point.
\begin{lemma}
    \label{sec_adaptive:claim_expected_score_at_turn}
    For every turn~$\turn$, any Guesser that plays against the~\adaptiveDealer is expected to score at most
    $$
    \fE\left[\correctTurnRV\right] \le O\left(\frac{1}{\minidecks}\right)
    .
    $$
\end{lemma}
\begin{proof}
    \label{sec_adaptive:proof_expected_score_at_turn}
    Let~$\stage$ be the stage prior to turn~$\turn$.
    Assume that~$\averageContributionRV = \averageContributionInst$ for some~$\averageContributionInst \in \averageContributionSupp$, and recall that~$\drawableMinidecks_\turn \subseteq [\minidecks]$ is the set of all drawable mini-decks at turn~$\turn$.
    Clearly, $\drawableAtLeast \subseteq \drawableMinidecks_\turn$, and therefore, the probability (over the randomness at turn~$\turn$) that a Guesser's guess is correct is at most:
    \begin{flalign}
    \label[ineq]{sec_adaptive:proof_expected_score_at_turn:probability_upper_bound}
    \Pr\left[\correctTurnRV | \averageContributionRV=\averageContributionInst\right]
    \le
    {1 \over |\drawableMinidecks_\turn|}
    \le
    {1 \over |\drawableAtLeast|}
    \le
    \frac{8}{\potentialParamSymbol} \cdot {\ln(\averageContributionInst) \over \minidecks}
    .
    \end{flalign}
    Where the last inequality follows from~\Cref{sec_adaptive:claim_number_of_drawable_minidecks_lb}.
    
    Thus, we expect turn~$\turn$ to yield at most~$
        \expectation\left[\correctTurnRV | \averageContributionRV=\averageContributionInst\right]
        \le
        \frac{8}{\potentialParamSymbol} \cdot {\ln(\averageContributionInst) \over \minidecks}
    $
    correct guesses. We get that
    \begin{flalign}
        \expectation\left[\correctTurnRV \right]
        &
        =
        \sum_{\averageContributionInst\in\averageContributionSupp}
        \expectation\left[\correctTurnRV | \averageContributionRV=\averageContributionInst\right]
        \cdot
        \Pr[\averageContributionRV=\averageContributionInst]
        \label[equa]{sec_adaptive:proof_expected_score_at_turn:by_total_expectation}
        \\
        &
        \le
        \sum_{\averageContributionInst\in\averageContributionSupp}
        \left[\frac{8}{\potentialParamSymbol} \cdot {\ln(\averageContributionInst) \over \minidecks}\right]
        \cdot
        \Pr[\averageContributionRV=\averageContributionInst]
        \label[ineq]{sec_adaptive:proof_expected_score_at_turn:from_the_discussion_above}
        \\
        &
        =
        \frac{1}{\minidecks} \cdot 
        \frac{8}{\potentialParamSymbol} \cdot 
        \expectation\left[
        \ln(\averageContributionRV)
        \right]
        \label[equa]{sec_adaptive:proof_expected_score_at_turn:by_defintion_of_expectation}
        \\
        &
        \le
        \frac{1}{\minidecks} \cdot 
        \frac{8}{\potentialParamSymbol} \cdot 
        \ln\left(
        \expectation\left[
        \averageContributionRV
        \right]
        \right)
        \label[ineq]{sec_adaptive:proof_expected_score_at_turn:from_jensen}
        \\
        &
        \le
        \frac{1}{\minidecks} \cdot 
        \frac{8}{\potentialParamSymbol} \cdot 
        \ln(O(1))
        \label[ineq]{sec_adaptive:proof_expected_score_at_turn:from_claim_expected_average_contribution_is_constant}
        =
        O\left(\frac{1}{\minidecks}\right)
        .
    \end{flalign}
    Where
    \Cref{sec_adaptive:proof_expected_score_at_turn:by_total_expectation} follows from the law of total expectation,
    \Cref{sec_adaptive:proof_expected_score_at_turn:from_the_discussion_above} is true by the discussion above,
    \Cref{sec_adaptive:proof_expected_score_at_turn:by_defintion_of_expectation} is true by definition of expectation and its linearity,
    \Cref{sec_adaptive:proof_expected_score_at_turn:from_jensen} follows from Jensen Inequality~\Cref{sec_preliminaries:thm_jensen_inequality},
    and~\Cref{sec_adaptive:proof_expected_score_at_turn:from_claim_expected_average_contribution_is_constant} follows from~\Cref{sec_adaptive:claim_expected_average_contribution_is_constant}.
\end{proof}

\paragraph{Summing up:} 
In \Cref{sec_adaptive:claim_number_of_minidecks} we showed that an \adaptiveDealer with~$m$ bits of memory tracks~$\minidecks=\Theta(m)$ mini-decks.
Thus, in~\Cref{sec_adaptive:claim_expected_score_at_turn} we showed that during the \phaseAdaptive, the expected score in each turn is at most~$O(1/m)$, and, by linearity of expectation, any Guesser is expected to score at most~$O(n/m)$ during the \phaseAdaptive.
By~\Cref{sec_adaptive:claim_expected_runtime_at_turn}, each of these turns is expected to take a constant time.

The main result for this section follows:
\begin{theorem}
    \label{sec_adaptive:thm_adaptive_during_adaptive_phase}
    An \adaptiveDealer with~$m$ bits of memory scores~$O(n/m)$ points in expectation against any Guesser during the \phaseAdaptive, and each draw takes~$O(1)$ time in expectation.
\end{theorem}

In~\Cref{sec_linear_perfect:sec_discussion} we improve upon this result and expand the analysis for all turns during the game.
We show that any Guesser scores~$O(n/m + \ln m)$ points in expectation, and that each turn runs in worst case constant time.
We do so by replacing the rejection sampling mechanism and addressing the~\phaseFinal.

\section{Sampling from a Dynamic Subset in Constant Time}

\label{sec_linear_perfect}

\newcommand{\intervalSymbol}{I}
\newcommand{\cellSymbol}{A}
\newcommand{\weightSymbol}{W}

\newcommand{\cellSize}{\log n}
\newcommand{\cellCount}{n / \log n}

\newcommand{\availableCards}{\texttt{available\_cards}}

\newcommand{\massQuotient}{q}
\newcommand{\massResidue}{r}
\newcommand{\totalMassQuotient}{\totalMass_\massQuotient}
\newcommand{\totalMassResidue}{\totalMass_\massResidue}

\newcommand{\population}{p}
\newcommand{\colors}{k}
\newcommand{\marbles}{m}
\newcommand{\urnDS}{U}

\newcommand{\mass}{a}
\newcommand{\gradualBound}{g}
\newcommand{\outcome}{p}
\newcommand{\graduallyChangingDS}{G}
\newcommand{\totalMass}{t}

In this section we present a data structure for maintaining and sampling from a dynamic subset in worst case constant time.
We use this data structure to complete the analysis of the \adaptiveDealer~(\Cref{sec_adaptive:alg_adaptive_dealer}) during the \phaseFinal, and to improve upon its runtime during the \phaseAdaptive. 
We also use it to construct an efficient Dealer that generates all permutations with equal probability; that is, one that draws available cards uniformly at random in each turn.
As a result, any Guesser that plays against this Dealer scores at most~$\ln n$ correct guesses in expectation.
The main point of this Dealer is that it requires~$O(n)$ bits of memory and runs in worst-case constant time.

Our constructions revolves around the ability to maintain and sample from the set of available cards. Generally speaking, for a given domain~$\domainSymbol$, we present a data structure for maintaining and sampling from a dynamic subset~$\subsetSymbol \subseteq \domainSymbol$.

\begin{theorem}
    \label{sec_linear_perfect:claim_maintain_and_sample_from_subset}
    There exists a data structure~$\subsetDS$ such that for any~domain~$\domainSymbol$, $\subsetDS$ maintains a dynamic subset~$\subsetSymbol \subseteq \domainSymbol$. The data structure~$\subsetDS$ supports membership queries, addition and removal of an element to and from~$\subsetSymbol$, and sampling of a random member of~$\subsetSymbol$ uniformly at random.
    All operations run in worst-case constant time, while using~$O(|\domainSymbol|)$ bits of memory.
\end{theorem}

Given such a data structure, implementing our linear-space and constant time Dealer is straightforward: use~$\subsetDS$ to maintain the set of available cards.
Begin by initializing ~$\subsetDS$ to track the entire deck of~$n$ cards.
In each turn, the Dealer samples uniformly at random an available card to draw, and removes it from the set.
Since both operations run in constant time, so does the Dealer.
All available cards are equiprobable, so the expected score is~$\ln n$.
The Dealer solely uses~$\subsetDS$, so it consumes~$O(n)$ bits of memory.

\begin{theorem}
    \label{sec_linear_perfect:thm_linear_dealer}
    There exists a Dealer with~$O(n)$ bits of memory that generates a permutation uniformly at random. Any Guesser is expected to score at most~$\ln n$ correct guesses against this Dealer. 
    Furthermore, this Dealer is efficient, and each turn takes worst-case constant time.
\end{theorem}

We describe our data structure across~\Cref{sec_linear_perfect:sec_cells_intervals_populations,sec_linear_perfect:sec_distribution_of_populations}.
We then show how to use it in order to implement the \phaseFinal of the \adaptiveDealer and how to make it run in worst case constant time, these are discussed in \Cref{sec_linear_perfect:sec_discussion}.

\subsection{Subset Data Structure}
\label{sec_linear_perfect:sec_cells_intervals_populations}

It seems natural to allocate a bit for every element in the domain~$\domainSymbol$ to indicate its membership to~$\subsetSymbol$, and store these bits in small cells that allow fast operations.
With simple bit-wise operations on words, we can mark a cell's bit as a member of~$\subsetSymbol$, discard it from~$\subsetSymbol$, and query its membership.
Fixing the cells in an array forms a bitmap that supports addition and removal of a single element, as well as set membership queries, all in constant time and within the memory requirements.

What about sampling?
We think of a cell as inducing a restricted subset of~$\subsetSymbol$.
The cells are small enough, so we can sample a random element from it in constant time.
But how do we sample a cell? 
Choosing one at random implies that elements from less populated cells are drawn with a higher probability than elements from less populated ones. 
It follows that to draw an element uniformly at random, we need to sample a cell proportionally to its population size, that is, to the number of elements from~$\subsetSymbol$ that it tracks.

Observe that cells that hold the same number of elements must be equally probable.
This hints that our data structure requires the ability to uniformly sample a cell of a given population size, and also requires the ability to sample a population size in proportion to the number of elements from that size.
So on a high-level, the data structure first samples a population size, then picks a cell with this many elements, and finally, chooses a random element from that cell.

For brevity, we denote the size of the domain~$\domainSymbol$ by~$n=|\domainSymbol|$, and assume that~$n$ is a power of~$2$.
We also assume that there is a natural mapping between~$[n]$ and $\domainSymbol$.

\paragraph{Cells:}
We split the domain~$\domainSymbol$ to cells, where each cell tracks~$\log n$ elements.
If a cell contains~$\population$ elements from~$\subsetSymbol$, then we say that the cell is of population~$\population$, and similarly, that the elements in it are of population~$\population$.

In terms of space, each cell stores~$\log n$ bits to track the membership of~$\log n$ elements.
We would like to move cells around, so each cell also contains an index~$j$ so the~$i$th bit of that cell indicates the membership of the element numbered~$(j \cdot \log n) + i$.
Overall, each cell requires~$O(\log n)$ bits, and there are~$n /\log n$ cells, so we are good.
Simple bit manipulation on words allows us to sample a random element and update its membership in constant time, as can be seen in~\Cref{sec_linear_perfect:alg_draw_card_from_cell}. 
But how do we sample a cell?

\paragraph{Intervals:}
After removing an element from a cell of population size~$\population$ it becomes a cell of population size~$\population-1$.
In other words, the set of cells of population~$\population-1$ grows by~$1$, at the expense of the set of cells of population~$\population$, while the union of both sets remains in the same size.
This works the same also for adding an element to a cell, and hints that updates, as well as sampling, can be done in a local manner.

We store the cells of the same population size sequentially and in an ordered manner, so that the cells of population~$\population$ are stored before those of population~$\population-1$.
We refer to a sequence of cells of the same population size by \emph{interval}.
For each interval, we store the number of cells of that interval, as well as a pointer to its first cell.
This allows us to maintain the ordered structure, as well as to sample a random cell of a given population size, at worst case constant time, while using only~$O(\log^2 n)$ additional space.
An algorithmic description is provided in~\Cref{sec_linear_perfect:alg_draw_cell_of_population}.
A visual example of interval maintenance is provided in~\Cref{sec_linear_perfect:fig_interval}.

\begin{figure}[ht]
  \caption{Cells ordered in intervals. Top to bottom:
  (i) there are~$8$ cells mentioned by their index~($j$). Cell~$4$ holds~$3$ elements (the blue interval), cells~$1,3,7$ and~$5$ contain $2$ elements (the green interval), cells~$0$ and~$6$ track a single element (the yellow interval), and cell~$2$ is empty (red).
  (ii) We remove an element from cell~$3$, thus, it switches places with the last cell in the interval, which is cell~$5$.
  (iii) Cell~$3$ holds a single element now, and therefore it is marked as the beginning of the yellow interval (cells of population~$1$).
  .\\}
  
  \centering
  \includegraphics[width=0.4\columnwidth]{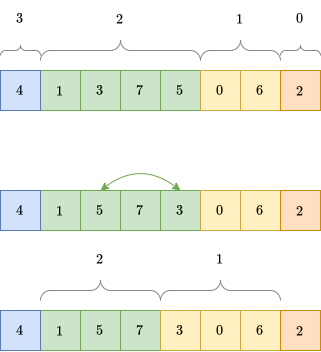}
  \label{sec_linear_perfect:fig_interval}
  
\end{figure}

\paragraph{Population:}

Since each cell contains at most~$\log n$ elements, then there are~$\log n$ possible population sizes (empty cells are not interesting).
When we remove an element from a cell of population~$\population$, that cell becomes one of population~${\population-1}$.
In other words, let~$\mass_\population$ be the number of elements that reside in cells that contain exactly~$\population$ elements.
Observe that after removing an element of population~$\population$, $\mass_\population$ is reduced by~$\population$, and~$\mass_{\population-1}$ grows by~$\population-1$.
We would like to sample a population size according to the dynamic pseudo distribution induced by~$$(\mass_1, \dots, \mass_{\log n}).$$
We claim this is possible with the given time and memory requirements, and show the proof in \Cref{sec_linear_perfect:sec_distribution_of_populations}.
\begin{lemma}
    \label{sec_linear_perfect:claim_population_sampling}
    It is possible to maintain the dynamic pseudo distribution induced by~$(\mass_1, \dots, \mass_{\log n})$ while using~$O(n)$ bits, where updates and sampling run in worst-case constant time.
\end{lemma}

Having settled the space and run time requirement, it is left to show that we actually sample each element of~$\subsetSymbol$ uniformly at random.
This completes the proof of~\Cref{sec_linear_perfect:claim_maintain_and_sample_from_subset}.

\begin{claim}
    \label{sec_linear_perfect:claim_equiprobable}
    All elements of~$\subsetSymbol$ are equiprobable.
\end{claim}
\begin{proof}
    \newcommand{\randomElement}{b}
    Fix some element~$\randomElement \in \subsetSymbol$ that resides in a cell of population~$\population$, whose interval contains~$\ell$ cells, and assume that~$|\subsetSymbol| = \alpha$.
    Let the indicator random variable $\mathbf{Y}_\randomElement$ be the event that the sampling process generated~$\randomElement$.
    The element~$\randomElement$ is sampled in probability~$1/\population$ if its cell was selected~(and~$0$ otherwise), which happens in probability~$1/\ell$ conditioned on the event that population~$\population$ was sampled~(and~$0$ otherwise). 
    The latter happens with probability~$\mass_\population / \alpha$.
    Since~$\mass_\population = \ell \cdot \population$, we get that $$
    \Pr[\mathbf{Y}_\randomElement] = {1 \over \population} \cdot {1 \over \ell} \cdot {\ell \cdot \population \over \alpha} = {1 \over \alpha}
    .
    $$
\end{proof}

\subsection{Pseudo Distribution of Populations}
\label{sec_linear_perfect:sec_distribution_of_populations}
\emph{And now for something completely different.}

We now present a data structure for maintaining and sampling from the dynamically changing distribution of populations\footnote{See~\Cref{sec_preliminaries:sec_dynamic_dist} for some background about dynamically changing distributions.}, and thus prove~\Cref{sec_linear_perfect:claim_population_sampling}.

To recap, and as a prelude to what's to come, we would like to sample a population with respect to the number of elements of that population.
This distribution is induced by the vector~$(\mass_1, \dots, \mass_{\log n}) \in [n]^{\log n}$, where~$\mass_\population$ is the number of elements that reside in cells that contain~$\population$ elements.
As mentioned, observe that the distribution of population changes in a specific way: 
removing an element from a cell of population~$\population$, implies that~$\mass_\population$ decreases by~$\population$, and~$\mass_{\population-1}$ increases by~$\population-1$, and vice versa when adding an element.
So for every~$\population$, $\mass_\population$ changes by at most~$\log n$.

Consider the presentation of a number~$\mass \in \fN$ as being made of an integer multiplication of~$\log n$ and some leftover; this results in a quotient and a residue.
\begin{flalign*}
    \mass
    =
    \underbrace{
    \left\lfloor
    {\mass \over \log n}
    \right\rfloor
    }_{\text{quotient}}
    \cdot
    \log n
    +
    \underbrace{
    (\mass \bmod \log n)
    }_{\text{residue}}
    .
\end{flalign*}
Observe now that the addition and subtraction of any quantity that is at most~$\log n$, changes the quotient at most by one. This fact will be of great use to us.

Given a pseudo probability mass vector~$(\mass_1, \dots, \mass_{\log n})$ we split the masses into quotients and residues and manage them in two different data structures.
Let~$\massQuotient_i = \lfloor {\mass_i / \log n}\rfloor$, and~$\massResidue_i = (\mass_i \bmod \log n)$ be the quotient and residue (resp.) of the mass of outcome~$i$.
Let~$\totalMassQuotient = \log n \cdot \sum_{i=1}^{\log n} \massQuotient_i$ be the total mass of quotients, and let~$\totalMassResidue = \sum_{i=1}^{\log n} \massQuotient_i$ be the total mass of residues.

The pseudo distribution of quotients~$(\massQuotient_1 \cdot \log n, \dots, \massQuotient_{\log n} \cdot \log n)$ is maintained in one data structure, and the pseudo distribution of residues~$(\massResidue_1,\dots,\massResidue_{\log n})$ in another. 
Flip a biased coin with probability proportional to the ratio between the total mass of quotients and the total mass of residues~$\totalMassQuotient: \totalMassResidue$, and sample from the resulting data structure.
This process clearly samples from the distribution~$(\mass_1, \dots, \mass_{\log n})$.
So, it is left to show how to maintain the two distributions.

\paragraph{Quotients:}
Observe that the quotients pseudo distribution induced by~$(\massQuotient_1 \cdot \log n, \dots, \massQuotient_{\log n} \cdot \log n)$ is exactly the pseudo distribution induced by~$(\massQuotient_1, \dots, \massQuotient_{\log n})$. 
And as discussed, each update adds or removes at most a single point of mass there.
This is a special case of a dynamically changing distribution that is used for simulating sampling of colored marbles from an urn, with and without replacement.

This can be done in a simple and explicit way\footnote{That might entertain some job interviewers.}.
To simulate an urn that contains~$\marbles$ marbles of~$\colors$ colors, we represent each color as a doubly-linked list, where each node in the list represents a single marble. By storing an anchor to each linked list, we get that addition and removal can be implemented in worst-case constant time.
As for sampling, the nodes of all linked lists are stored in an array of size~$\marbles$.
We make sure that there are no empty entries, so when a marble is removed, we fill its position in the array with another marble and fix its pointers.
So sampling a marble is done by sampling a random index in the array and checking its color.
Sampling without replacement is a combination of sampling and removal.

An algorithmic description of the Urn data structure is provided in~\Cref{sec_linear_perfect:alg_urn_ds} and~\Cref{sec_linear_perfect:alg_urn_methods}.
A visual representation of the removal process is provided in~\Cref{sec_linear_perfect:fig_marble_removal}.

\begin{figure}[ht]
  \caption{Urn data structure - marble removal process, top to bottom: (i) Initially, there are~8 marbles of~3 colors. We would like to remove a blue marble. (ii) The node in index~$2$ is removed, and the pointers are updated accordingly. (iii) To preserve the array's continuity, we move the last marble in the array and store it in the absent place, while fixing the pointers. }
  
  \centering
  \includegraphics[width=0.4\columnwidth]{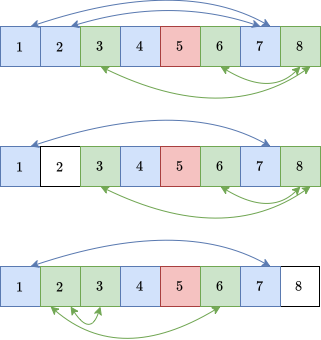}
  \label{sec_linear_perfect:fig_marble_removal}
  
\end{figure}

\begin{claim}
    \label{sec_linear_perfect:claim_urn}
    There exists a data structure~$\urnDS$ that stores an urn of~$\marbles$ marbles, where each marble is colored in one of~$\colors$ colors, such that
    $\urnDS$ requires~$O\left(\marbles \log \marbles + \marbles \log \colors\right)$ bits of memory, 
    and supports sampling, addition and removal of a marble in worst-case constant time.
\end{claim}

It follows that we can maintain the quotients distribution by using the Urn data structure with~$\marbles={n \over \log n}$ and~$\colors=\log n$, for which updates and samples will run in worst-case constant time, and within the given space requirements.

\paragraph{Residues:}
We observe that the total mass of residues~$\totalMassResidue$ is less than~$\log^2 n$, and that there are~$\log n$ possible outcomes.
It is possible to maintain and sample from such distributions by using a construction by Hagerup, Mehlhorn and Munro~\cite{hagerup1993} as described in~\Cref{sec_preliminaries:sec_dynamic_dist:thm_polynomially_bounded}.
This would give an expected constant sampling run time.

\newcommand{\randomPrefix}{u}
Another option is to observe that while there are more than~$\log^3 n$ elements, we sample from the distribution of residues with probability~$O(1/ \log n)$. So we can allow ourselves a relatively long sampling time.
One way to do so is to store the masses explicitly, and when asked to generate an outcome, we first sample an integer~$\randomPrefix \sim[\sum_{i=1}^{\log n} \massResidue_i]$ uniformly at random, and compare it against the corresponding prefix vector of the pseudo distribution, that is~$(0, \massResidue_1, \massResidue_1+\massResidue_2, \dots, \sum_{i=1}^{\log n-1} \massResidue_i)$, and return the index of the largest entry that is smaller than~$\randomPrefix$.
Since this can be done in~$O(\log n)$ time, we get constant time in expectation.
But we aim for worst-case constant time, so we will do better.

\newcommand{\samplingFunction}{f}
\newcommand{\supportSize}{k}
\newcommand{\massBound}{b}

Recall that we have~$O(n)$ bits of memory, and can perform a set of operations on words of size~$\log n$ (See~\Cref{sec_preliminaries:sec_word_ram}).
As a result, we can explicitly store a small function in memory, so that evaluating it takes only a constant time, say, by storing an array~$A$ such that~$A[x]$ holds the evaluation of the function on input~$x$.

Call a pseudo distribution \emph{$\massBound$-mass bounded} if the mass of every outcome is at most~$\massBound$.
Consider a function~$\samplingFunction$ that receives a vector~$(\mass_1, \dots, \mass_\supportSize) \in [\massBound]^\supportSize$ and a number~$\randomPrefix \in [\massBound \cdot \supportSize]$ and outputs a member of~$[\supportSize]$ as described in the prefix-sum scheme above, that is~$\samplingFunction : [\massBound]^\supportSize \times [\massBound \cdot \supportSize] \to \supportSize$.
By choosing a random~$\randomPrefix$ uniformly at random from~$[\sum_{i \in [\supportSize]} \mass_i]$ we, in fact, sample according to the given pseudo distribution vector.
What does it take to store $\samplingFunction$ in memory?

\begin{claim}
    \label{sec_linear_perfect:claim_maintaining_small_distributions}
    For every~$\massBound,\supportSize \in \fN$,
    if~$\massBound^{\supportSize} \cdot (\massBound \cdot \supportSize) \cdot \log \supportSize \le n$, then it is possible to maintain, and sample from, any dynamically changing $\massBound$-mass bounded pseudo distribution with support of size~$\supportSize$ in worst case constant time.
\end{claim}
\begin{proof}
    Store the vector of masses explicitly and the sum of masses; this takes~$O(\supportSize \cdot \log \massBound)$ bits and allows updates in constant time.
    There are~$\massBound^\supportSize$ possible $\massBound$-mass bounded distribution vectors.
    For each one of them, we sample a random~$\randomPrefix$ such that~$\randomPrefix \le \sum_{i \in [\supportSize]} \mass_i \le \supportSize \cdot \massBound$.
    So overall, we can store the sampling function~$\samplingFunction$ as an array of~$\massBound^{\supportSize} \cdot (\massBound \cdot \supportSize)$ entries (one for every possible input), where each entry contains outcome, which is of size~$\log \supportSize$. 
    Now sampling takes a constant time, and our only wish is that this array would fit into memory, which is given by the assumption.
\end{proof}

Back to the distribution of residues, while being~$\log n$-mass bounded, the support is too large for the sampling function to fit into memory.
We overcome this obstacle by partitioning the support of the distribution to smaller sets of equal size, and maintaining the pseudo distribution of each one of them separately.
Now we are left with the task of maintaining and sampling from the distribution of support sets.
This can be visualized as a depth-$2$ tree, where each leaf is responsible for maintaining a distribution of a certain partition, and has a total mass, where the root maintains the distribution of the leaves with respect to their total mass.
So, generating an outcome is done by first sampling a leaf according to the distribution at the root, and then sampling an outcome according to the distribution at the selected leaf.
Updates involve only the relevant leaf and the root, so they are both performed at constant run time.
In terms of space, this requires~$O(\log n \cdot \log \log n)$ bits of memory, in addition to storing the sampling functions in memory.
Since the distributions in the leaves are similarly mass bounded and are over the same number of outcomes, it suffices to store a single sampling function for the leaves (call it~$\samplingFunction_0$), and another one for the distribution at the root~(call it~$\samplingFunction_1$).

In more detail, we partition~$[\log n]$ to~$\supportSize_1 = 2 \log \log n$ disjoint sets, each of size~$\supportSize_0 = \frac{\log n}{2 \log \log n}$, and we consider the pseudo distributions restricted to each support set.
Observe that the distribution of any particular support set is~$\massBound_0$-mass bounded where~$\massBound_0 = \log n$.
As a result, the total mass of any leaf is upper bounded by~$\massBound_1 = \massBound_0 \cdot \supportSize_0 = \frac{\log^2 n}{2 \log \log n}$.
It follows that the distribution at the root is~$\massBound_1$-mass bounded, and has~$\supportSize_1$ possible outcomes.
Similarly, the distributions at the leaves are all~$\massBound_0$-mass bounded, and have support of size~$\supportSize_0$.
Let~$\samplingFunction_0$ ($\samplingFunction_1$) be the function for sampling from any $\massBound_0$-mass bounded (resp. $\massBound_1$-mass bounded) distributions over support of size~$\supportSize_0$ (resp. $\supportSize_1$).

\begin{claim}
    It is possible to store~$\samplingFunction_0$ and $\samplingFunction_1$ in memory.
\end{claim}
\begin{proof}
    To use~\Cref{sec_linear_perfect:claim_maintaining_small_distributions} we would like to make sure that the inequality
    \begin{flalign*}
        \underbrace{\massBound^{\supportSize}}_{(\ast)} 
        \cdot 
        \underbrace{(\massBound \cdot {\supportSize}) \cdot \log \supportSize}_{(\ast\ast)} \le n
    \end{flalign*}
    holds for~$(\massBound, \supportSize) = (\massBound_0, \supportSize_0) = (\log n, \frac{\log n}{2 \log \log n})$ and for~$(\massBound, \supportSize) = (\massBound_1, \supportSize_1) = (\frac{\log^2 n}{2 \log \log n}, 2 \log \log n)$.
    We observe that for both cases~$(\ast\ast)$ is upper bounded by~$\poly\log n$. We show that~$(\ast) \le \sqrt{n}$, and the result follows.
    
    And indeed, for~$\massBound_0$ and~$\supportSize_0$:
    \begin{flalign*}
        \supportSize_0 \cdot \log \massBound_0
        &
        =
        \frac{\log n}{2 \log \log n} \cdot \log\log n 
        = 
        \frac{1}{2} \log n
        .
    \end{flalign*}
    Taking exponents on both sides implies that $(\ast) = {\massBound_0}^{\supportSize_0} = \sqrt{n}$.

    As for~$\massBound_1$ and~$\supportSize_1$:
    \begin{flalign*}
        \supportSize_1 
        &
        = 
        2 \log \log n
        \le 
        \frac{\log n}{4 \log \log n}
        =
        \frac{\log n}{2 \log \left(\log^2 n\right)}
        \le
        \frac{\log n}{2 \log \left(\frac{\log^2 n}{\log \log n}\right)}
        =
        \frac{\log n}{2 \log \massBound_1}.
        \\
        \supportSize_1 \cdot \log \massBound_1 & \le \frac{\log n}{2}.
        \\
        \massBound_1^{\supportSize_1} &\le \sqrt{n}
        .
    \end{flalign*}
    
\end{proof}
As a corollary, updating the distribution of residues, as well as sampling from it takes only a constant time in the worst case, and requires~$O(\log n \cdot \log \log n)$ bits in addition to two preprocessed functions that takes~$O(n)$ bits.

Since it is possible to maintain and sample in constant time, both from the distribution of quotients and the distribution of residues, both within the given space requirements, then it is possible to maintain and sample from the distribution of populations in worst-case constant time.
\Cref{sec_linear_perfect:claim_population_sampling} follows.

\subsection{Back to the Adaptive-Threshold Dealer}
\label{sec_linear_perfect:sec_discussion}

Looking back at our \adaptiveDealer from \Cref{sec_adaptive:alg_adaptive_dealer}, we would like to resolve two issues: (i) the time per draw during the~\phaseAdaptive, making it constant time in the worst case rather than expected case and (ii) to describe the Dealer's internals for the \phaseFinal.

Starting with the~\phaseAdaptive, in each turn, \adaptiveDealer samples a mini-deck repeatedly until a drawable one is found.
We showed (in \Cref{sec_adaptive:claim_expected_runtime_at_turn}) that the process of finding a suitable mini-deck is expected to terminate after a constant time.
We claim that there is a better way.
Consider the set~$\drawableMinidecks_\turn \subseteq [\minidecks]$ of all mini-decks that contain at least one hole at turn~$\turn$, thus, we can draw a card from them.
At the beginning of each stage, all mini-decks are such, and when a mini-deck reaches the threshold, we remove it from this set.
We can use the data structure from~\Cref{sec_linear_perfect:claim_maintain_and_sample_from_subset} to maintain and sample from the set of drawable mini-decks in worst-case constant time, and within the required space.
We keep two copies of this data structure: one for the odd stages and one for the even stages. In each stage, we maintain and sample from one of them as described above, while we replenish the other to be ready for the beginning of the next stage.
This implies that the \adaptiveDealer runs in worst-case constant time per turn during the~\phaseAdaptive.
Since each drawable mini-deck is equiprobable, the analysis of the \phaseAdaptive holds, and we get the following improvement upon \Cref{sec_adaptive:thm_adaptive_during_adaptive_phase}.
\begin{lemma}
    \label{sec_linear_perfect:lemma_adaptive_dealer_during_adaptive_phase}
    An \adaptiveDealer with~$m$ bits of memory scores~$O(n/m)$ points in expectation against any Guesser during the \phaseAdaptive, and each draw runs in worst case constant time.
\end{lemma}

As for the~\phaseFinal, recall that in \Cref{sec_adaptive:alg_adaptive_dealer}, the \adaptiveDealer draws the last~$2\minidecks$ cards uniformly at random.
That is, we would like to generate a permutation of the remaining cards.
Observe that generating a permutation of~$2\minidecks$ cards is a similar task to what we did in this section.
On the other hand, since a mini-deck might hold all remaining cards, then the initial super-set of possible cards is now~$\minidecks \cdot 2\minidecks$, which requires more bits than we are willing to spend.
So we need to come up with an efficient way to enumerate the remaining available cards, that is, to easily associate the set~$[2\minidecks]$ with the set of available cards.

We introduce a small enhancement to the Cells data structure (\Cref{sec_linear_perfect:sec_cells_intervals_populations}) that allows us to draw the last~$2\minidecks$ remaining cards in worst-case constant time, and in a perfectly random way.
This implies that any Guesser is expected to score at most~$\ln(2 \minidecks)$ points during this phase.

Originally, we split the domain into consecutive sets of size~$\log n$ and tracks each set in a cell, where a cell maintains a bitmap, as well as some meta-data to recover the original identity of each bit.
The idea now is that we use the data structure to track the set of available cards. Each cell tracks cards of (possibly) multiple mini-decks, in a way that allows us to associate each availability bit with a mini-deck. And complementary, if a mini-deck has more available cards than a cell can handle, then we can track it across multiple cells.

In the enhanced version, each cell tracks~$\log \minidecks$ cards.
We iteratively assign mini-decks to cells.
If a cell can track what's left of a mini-deck entirely, then we just add it to the cell. 
Otherwise, we assign what portion of the mini-deck that does fit into the cell and continue to the next cell.
We begin initializing this data structure ahead of the~\phaseFinal while ensuring worst-case constant time (say, start in stage~$n-4$ and inspect~4 mini-decks, or portions of mini-decks, in each turn).
By starting early, we ensure that no mini-deck is empty, which simplifies this construction.

In terms of metadata and space, each cell now contains:
\begin{itemize}
    \item Cards availability bitmap ($\log \minidecks$ bits) - that tracks the availability of~$\log \minidecks$ cards.
    \item First card bitmap ($\log \minidecks$ bits) - refers to the Cards availability bitmap, and indicates whether a card is the first card of a mini-deck.
    \item The first mini-deck assigned to this cell~($\log \minidecks$ bits).
    \item The first mini-deck's offset~($O(\log \minidecks)$ bits) - for the case that some of the cards of the first mini-deck are tracked in previous cell(s).
\end{itemize}
See~\Cref{sec_linear_perfect:fig_adaptive_cells} for a visual representation.

\begin{figure}[h]
  \caption{Enhanced cells for the~\adaptiveDealer. 
  The left cell tracks mini-decks~$14,15$ and the first~$4$ cards of mini-deck~$16$. It is of population~$4$ since it tracks a total of~$4$ available cards.
  Mini-deck~$14$ is the first mini-deck that begins in this cell, thus, it is a leader.
  The right cell tracks the~$6$ cards left of mini-deck~$16$ and the first two cards of mini-deck~$17$. It is of population~$6$.
  Mini-deck~$17$ is a leader because the first card of mini-deck~$16$ is tracked in the left cell.
  \\}
  
  \centering
  \includegraphics[width=1.0\columnwidth]{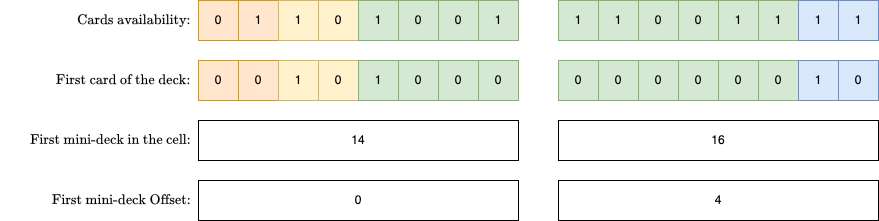}
  \label{sec_linear_perfect:fig_adaptive_cells}
  
\end{figure}

How many cells do we need? 
Since there are~$O(\minidecks)$ available cards to track, we get that~$O(\minidecks/ \log \minidecks)$ cells suffice. 
Since each cell consumes~$O(\log \minidecks)$ bits, then we are good.

Observe that during initialization, we maintain two data structures for the same purpose, so there is some bookkeeping to perform.
We add mini-decks to cells sequentially, so we make sure to know what is the last fully added mini-deck.
When a drawable mini-deck is chosen, if it is fully tracked, we sample from it using the cells data structure. And if not, then we draw its top card and update its hole counter.
Once a mini-deck is fully added, we freeze its hole counter, and use it later for decoding. We need to pay extra attention to the currently added mini-deck.

To draw a card from a tracked mini-deck during initialization, we need to recover its cell. We introduce yet another data structure for this task.
We say that a mini-deck is \emph{leading} if it is the first mini-deck that begins in a cell. 
So in~\Cref{sec_linear_perfect:fig_adaptive_cells}, mini-decks~$14$ and~$17$ are leading.
We allocate a bit for each mini-deck to indicate whether it is leading or not ($\minidecks$ bits).
Additionally, for leading mini-decks, we store a pointer to their cell.
So drawing a card from a leading mini-deck is easy.

Looking at the leaders bitmap, observe that every sequence of~$\log \minidecks$ bits indicates at least one leading mini-deck.
This is true because, during initialization, each mini-deck contains at least one available card.
So we find the cell of a non-leading mini-deck by finding its leader, which requires reading~$\log \minidecks$ bits from the leader's bitmap.
Since there are~$O(\minidecks / \log \minidecks)$ cells, then there are at most that many leaders, so storing their pointers also fits into memory.

After initialization, that is, during the \phaseFinal, we use the same procedure of sampling a population size, then a cell of this size, and then an available card.
The ``decoding'' of each card's meaning is done, as previously, by simple bit-wise and arithmetic operations.
\begin{lemma}
    \label{sec_linear_perfect:lemma_adaptive_dealer_during_final_phase}
    An \adaptiveDealer with~$m$ bits of memory scores~$O(\ln m) $ points in expectation against any Guesser during the \phaseFinal, and each draw runs in worst-case constant time.
\end{lemma}

Together with~\Cref{sec_linear_perfect:lemma_adaptive_dealer_during_adaptive_phase}, we finish the analysis of the \adaptiveDealer.

\begin{theorem}
    An \adaptiveDealer with~$m$ bits of memory scores~$O(n/m + \ln(m))$ points in expectation against any Guesser, and each draw takes worst-case constant time.
\end{theorem}

\paragraph{Time Complexity of Sampling at Random: }

In the literature of dynamically changing distributions, several solutions achieve expected constant time for certain families of distributions.
This is usually achieved by using some form of rejection sampling, that is, sampling repeatedly until reaching some desired outcome.
While none of our machinery was exceptionally complicated, it took some work to get to worst-case constant time, let alone do so in a space-efficient manner.

We do, however, rely on the ability to sample in constant time a random number from a set~$[x]$ for any $x\le 2^w$, where~$w$ is the word size in bits.
This is a non-trivial ability, and without it, the expected constant time sneaks back in.

When~$x$ is a power of two, say, $x=2^k$, this is easy, just read~$k$ random bits and the outcome is the result.
But what happens when this is not the case?
One common solution is to repeatedly sample the required number of bits, until we get a result within the desired range. 
Since the sampled range is at most twice as large as the desired range, then the stopping time is distributed geometrically, and we get that this process runs in constant time \emph{in expectation}.
Another option is to approximate the distribution by reading a large number of random bits, say~$z~[2^y]$, and return~$\lfloor z / x \rfloor$. This process yields an additive approximation error of at most~$2^{-y}$.

Note that 
Feldman et al.~\cite{Feldman1993} studied this problem and showed that there exists a set of~$\Theta(w)$ values of predefined biases, such that for any~$x \le 2^w$, it is possible to sample uniformly at random from~$[x]$ given biased coin flips from this set. 
Despite the small number of random bits and the small number of predefined biases, their construction requires arithmetic computations that have no known efficient algorithm.

\section{Tight Lower Bound on the Predictability for a Memory Size}
\label{sec_lower_bound}
We now provide a lower bound on the predictability of the Dealer's permutation as a function of the memory it has. 
Specifically, we show that for any Dealer with $m$ bits of memory, there is a Guesser that makes at least $\Omega(n/ m + \ln (m)) $ correct predictions in expectation when playing against that Dealer. We call this Guesser the {\em Myopic Optimizing Guesser}.

We wish to show that the entropy of the generated permutation is relatively small for any low-memory Dealer. 
Observe that our low-memory ($m$ bits) Dealer (the \adaptiveDealer described in~\Cref{sec_adaptive}) could draw from a set of possibilities of at most $m$ cards in each turn. We show that this is not coincidental.
We do so by presenting an encoding scheme for the series of choices made by the Dealer. 
In each turn, the Dealer has a set of choices from which a card is chosen, according to some probability.
We utilize the fact that if the entropy of a random choice is low, then it is predictable, and a simple Guesser can guess correctly with some probability.

\newcommand{\encodeChoicesFuncState}{\tn{Encode}_{\dealer, \memoryState, \turn}}

\newcommand{\choicesRV}{\{\choiceRV_i\}_{i=1}^m}
\newcommand{\choicesInst}{\{\choiceInst_i\}_{i=1}^m}
\newcommand{\choicesnRV}{\{\choiceRV_\turn\}_{\turn=1}^n}
\newcommand{\choicesnInst}{\{\choiceInst_\turn\}_{\turn=1}^n}
\newcommand{\choicesByTurnInst}{\{\choiceInst_i\}_{i=1}^{\turn-1}}
\newcommand{\entropyConditionalChoiceTurn}{\entropy{\choiceRV_\turn|\choiceRV_1, ..., \choiceRV_{\turn -1}}}

To provide a bound on the expected number of correct guesses, we introduce an encoding scheme for the sequence of choices that the Dealer made during the game.
Our encoding scheme is made of two ingredients: 
We first observe that given two memory states at different turns, we can deduce which cards were drawn between them: this is an inherent property of any Dealer that draws each card exactly once.
By specifying the card drawing order, we reconstruct the exact choices made by the Dealer at each turn.
This allows us to fully recover the course of the game from the Dealer's point of view.
	
Consider the memory state~$\memoryState$ at some turn, and assume that the Dealer can draw from a set of~$k$ cards at that turn, such that making the~$i$th choice would lead the Dealer to state $\memoryState^i$ at the following turn.
Needless to say,~$\{\memoryState^i\}_{i=1}^{k}$ must all be distinct, since otherwise there are two choices that lead to the same memory state, and the Dealer would not be able to tell which choice was made, thus would not know which cards still resides in the deck.
It follows that if we know $\memoryState$, and we are given $\memoryState^i$, then we can tell exactly which card was drawn by the Dealer.
	
This argument generalizes to multiple turns.
Given two memory states, $\memoryState_1, \memoryState_2$ at turns~$\turn_1$ and~$\turn_2$ (respectively), such that there is a set of choices that leads the Dealer from $\memoryState_1$ to $\memoryState_2$, the set of cards drawn by the Dealer between $\turn_1$ and $\turn_2$ is determined.
The definition of range below captures this notion.

\newcommand{\rangetext}{Range\xspace}
    
\begin{definition}[\rangetext]
    \label{sec_lower_bound:def_range}
    The \rangetext of two memory states~$\memoryState_1, \memoryState_2$ between turns~$\turn_1$ and~$\turn_2$, denoted by~$\rangeFunction_{\turn_1, \turn_2}(\memoryState_1, \memoryState_2)$, is the set of cards that the Dealer draws between turns~$\turn_1$ and~$\turn_2$ if the Dealer's memory state was~$\memoryState_1$ at turn~$\turn_1$ and~$\memoryState_2$ at turn~$\turn_2$.
    \end{definition}
By the discussion above, we conclude that the \rangetext is well defined. However, while the {\em set} of cards generated between two memory states is unique, there may be multiple choice sequences that lead from one memory state to another.
By specifying the drawing order, we can reconstruct the exact choices made by the Dealer at each turn.
	
We are interested in bounding the (un)predictability of the choices made by the Dealer. 
To do so, we present an encoding scheme that utilizes the above properties. The encoding scheme simulates and records a game between the Guesser and the Dealer, and stores sufficient information to recover the entire course of the game.
In particular, we encode:

\begin{itemize}
\item
Snapshots of the Dealer's memory, taken at equal intervals: to recover the \rangetext.
\item 
Permutations, for recovering the order of the elements in each \rangetext.
\end{itemize}
	
\begin{figure}[htbp]
  \caption{Snapshots taken every~$m$ turns}
  \centering
  \includegraphics[width=0.9\columnwidth]{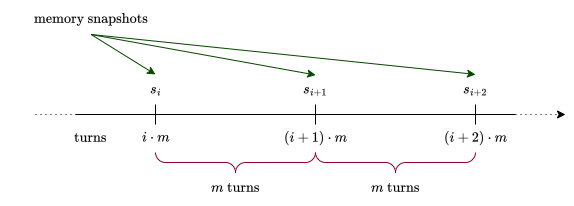}
  \label{sec_lower_bound:fig_memory_snapshots}
\end{figure}

Fix an encoding scheme for permutations where every permutation is represented by the same length.  (see~\Cref{storing_permutations} or~\cite{BonehCohen2023} for a discussion).

\begin{construction}
\label{sec_lower_bound:def_encoding}
To encode an ordered sequence of $n$ choices $\choicesnInst$ taken by the Dealer throughout the game, the function $\encodeFunction$ simulates and records a \dealerText playing a game while making the choices $\choicesnInst$.

For every $i \in [n /m]$, let $\memoryState_i$ be the Dealer's memory state at turn $m \cdot i$. 
Let $$
R_i = \rangeFunction_{m \cdot (i-1), m \cdot i}(\memoryState_{i-1}, \memoryState_{i})
$$
be the \rangetext of the cards drawn during the turns $\left\{m \cdot (i-1), \dots, (m \cdot i) -1 \right\}$. 
Enumerate the cards of $R_i = \{c_{i,1}, \dots, c_{i,m} \}$ such that $c_{i,j} < c_{i, j+1}$.
Let $\permutation_i \in \permutations_m$ be the permutation describing the order in which cards from~$R_i$ are drawn; that is, the card $c_{i, j}$ is drawn at turn $m \cdot (i-1) + \permutation_i(j)$.
		
The encoding function stores:
\begin{itemize}
\item 
    $n /m $ memory states $\{\memoryState_{i}\}_{i=1}^{n/m}$ ($ m \cdot (n/m)$ bits),
\item
    $n/m$ corresponding permutations $\{ \permutation_i \}_{i=1}^{n/m}$ ($(m \log m) \cdot (n/m)$ bits).
\end{itemize} 
Overall, the output's length is~$n + n\log m$ bits.
\end{construction}
A visual description of the snapshots taken in this construction is provided in~\Cref{sec_lower_bound:fig_memory_snapshots}. 

Consider the decoding process that gradually recovers choices one after the other while taking the information decoded so far into account.
Clearly, this process cannot produce more information than is poured into it.
Since the decoder reads~$n + n \log m$ bits, we get that the amount of information gained, for any sequence, is at most~$n + n \log m$ bits.

Let random variable~$\choiceRV_\turn$ be a random choice taken by the Dealer at turn~$\turn$, and let~$\choicesnRV$ be the sequence of random choices taken by the Dealer throughout the game.

\begin{claim}
    \label{sec_lower_bound:claim_sum_of_conditional_entropy}
    For every Dealer with~$m$ bits of memory and every sequence~$\choicesnInst$, the sum of choice entropy conditioned on past choices is at most
    $$
    \sum_{\turn=1}^n \entropy{\choiceRV_\turn | \choiceRV_1 = \choiceInst_1, \dots, \choiceRV_{\turn-1} = \choiceInst_{\turn-1}}
    \le 
    n + n \log m
    .
    $$
\end{claim}

\newcommand{\conditionalEntropyInst}{x}
\newcommand{\conditionalEntropyRV}{\mathbf{X}}
\newcommand{\conditionalEntropyTurnTerm}{\entropy{\choiceRV_\turn | \choiceRV_1 = \choiceInst_1, \dots, \choiceRV_{\turn-1} = \choiceInst_{\turn-1}}}
\newcommand{\turnSupp}{\cT}

\noindent
{\bf The Myopic Optimizing Guesser:}
The Guesser we suggest is simply the one that at every point maximizes the probability of guessing correctly the next move by the Dealer, given the Dealer and the history so far. If there is a tie (two values that have the same probability), then the card with the smaller face value is chosen.  
Note that this Guesser is deterministic.
Also note that this is not necessarily the best Guesser against the given Dealer since the Dealer may be influenced by the guesses it sees, so optimizing the next step does not necessarily mean optimizing the overall rate of prediction.

\Cref{sec_preliminaries:lemma_entropy_heavy_element} states that an upper bound on the entropy of a random variable (for example, a choice) implies a lower bound on the probability mass of the heaviest element in the support.
That is, whenever the entropy of a random choice is upperbounded, there is a guess that is guaranteed to be correct with probability inversely proportional to the entropy.
This sets a lower bound on the expected benefit of our Myopic Guesser.

\newcommand{\dealersRandomness}{\cR}

Let~$\correctMyopicRV$ be an indicator random variable for the event that the Myopic Guesser guessed correctly at turn~$\turn$.
Let~$\dealersRandomness$ be the distribution over the Dealer's randomness.
\begin{claim}
    \label{sec_lower_bound:lemma_correct_expectation}
    For every turn~$\turn$, for every Dealer and every sequence of~$\turn-1$ choices~$\choicesByTurnInst$, and any number~$y \ge 0$,
    if the entropy of the~$\turn$-th choice, conditioned on past choices, is at most~$y$, then the expected benefit of the Myopic Guesser at turn~$\turn$ is at least
    $$
    \expected_\dealersRandomness\left[\correctMyopicRV | \choiceRV_1 = \choiceInst_1, \dots, \choiceRV_{\turn-1} = \choiceInst_{\turn-1}\right] \ge 2^{-y}
    .
    $$
\end{claim}	
\begin{proof}
    If~$\conditionalEntropyTurnTerm \leq y$ then, from~\Cref{sec_preliminaries:lemma_entropy_heavy_element}, we get that there is a choice that the Dealer takes with probability at least~$2^{-y}$.
    Thus, the Myopic Dealer predicts correctly with probability at least~$2^{-y}$ and the result follows.
\end{proof}

\begin{theorem}
\label{sec_lower_bound:thm_lower_bound}
For every Dealer with~$m$ bits of memory, the Myopic Guesser scores at least~${n/ 2m}$ points in expectation.
\end{theorem}
\begin{proof}
    For any sequence of choices~$\choicesnInst$,
    and for every turn~$\turn$, denote the entropy of the~$\turn$-th choice conditioned on the first~$\turn-1$ turns by
    $$
    \conditionalEntropyInst_\turn = \conditionalEntropyTurnTerm,
    $$
    and let the random variable~$\conditionalEntropyRV$ be the entropy of a choice in a random turn conditioned on the past choices.
    \Cref{sec_lower_bound:claim_sum_of_conditional_entropy} implies that~$\expectation_{\turnSupp}\left[\conditionalEntropyRV\right] \le 1+ \log m$, where~$\turnSupp$ is the uniform distribution over~$n$ turns.

    Now consider a game played by the Myopic Guesser, and recall that in turn~$\turn$, the Guesser makes a prediction based on the~$\turn-1$ prior card draws, so if the first~$\turn-1$ choices were~$\choiceInst_1, \dots, \choiceInst_{\turn-1}$, then expected benefit at turn~$\turn$, is~$\expectation_{\dealersRandomness} \left[\correctMyopicRV | \choiceRV_1 = \choiceInst_1, \dots, \choiceRV_{\turn-1} = \choiceInst_{\turn-1}\right]$.

    Fix the function~$g(x) = 2^{-x}$.
    By linearity of expectation, we get that the expected score is at least

    \begin{flalign}
    \sum_{\turn \in [n]}
    \expectation_{\dealersRandomness} \left[\correctMyopicRV | \choiceRV_1 = \choiceInst_1, \dots, \choiceRV_{\turn-1} = \choiceInst_{\turn-1}\right]
    &
    \ge
    \sum_{\turn \in [n]}
    g(\conditionalEntropyInst_\turn)
    \label[ineq]{sec_lower_bound:prf_lower_bound:from_lemma_correct_expectation}
    \\
    &
    =
    n \cdot
    \expectation_{\turnSupp} \left[g(\conditionalEntropyRV)\right]
    \label[equa]{sec_lower_bound:prf_lower_bound:defintion_of_conditionalEntropyRV}
    \\
    &
    \ge
    n \cdot
    g\left(\expectation_{\turnSupp} \left[\conditionalEntropyRV\right]\right)
    \label[ineq]{sec_lower_bound:prf_lower_bound:from_jensen}
    \\
    &
    \ge
    n \cdot
    2^{-1-\log(m)}
    \label[ineq]{sec_lower_bound:prf_lower_bound:by_discussion_above}
    =
    {n \over 2m}
    .
\end{flalign}
Where
\Cref{sec_lower_bound:prf_lower_bound:from_lemma_correct_expectation} follows from~\Cref{sec_lower_bound:lemma_correct_expectation},
\Cref{sec_lower_bound:prf_lower_bound:defintion_of_conditionalEntropyRV} is true by the definition of~$\conditionalEntropyRV$,
\Cref{sec_lower_bound:prf_lower_bound:from_jensen} follows from Jensen Inequality (see ~\Cref{sec_preliminaries:thm_jensen_inequality}),
and~\Cref{sec_lower_bound:prf_lower_bound:by_discussion_above} follows from~\Cref{sec_lower_bound:claim_sum_of_conditional_entropy} and the discussion above.

\end{proof}

\section{Computationally Efficient Guessers: Open Book Dealers and Crypto to the Rescue?}
\label{prf}

The Myopic optimizing Guesser we saw in~\Cref{sec_lower_bound} is {\em not} necessarily computationally efficient {\em even in the case where the Dealer is efficient}: the Guesser has to find the most probable next move, and this is a computationally non-trivial task. Furthermore, if one-way functions exist, then it is a hard problem. We first point out that for any open-book Dealer there is an efficient Guesser.

\noindent
{\bf Open Book Dealer:}
When the Dealer has no secrets (as is the case for both of our suggested Dealers), then we can make the Myopic Optimizing Guesser of the previous section efficient. 
Recall that the Guesser has to figure out the most probable move by the dealer, which may require a long computation time.  
On the other hand, finding a move whose probability is an additive approximation of the most probable move, say within~$1/n$, is good enough: with this approximation, at every turn, instead of the probability of the original algorithm, we need to subtract~$1/n$, and this only modifies the expected result by~$n \cdot 1/n =1$.
Such a task can be performed efficiently: simply simulate the Dealer on the public state a few times ($O(n^2)$), figure out the most common response, guess that it is the most probable one, and act accordingly. 
Therefore, this Guesser is as efficient as the Dealer (times a polynomial in $n$ factor).

\begin{corollary}
\label{corrolary: open book}
For every open-book Dealer using $m$ bits of memory, there exists a Guesser that makes at least ${n \over 2m} - 1$ correct guesses in expectation and operates in the same amount of time as the Dealer times a polynomial in $n$ factor.
\end{corollary} %

\subsection{Space Efficient Generation Under Computational Assumptions}
	\label{one-way}
	
Suppose that one-way functions exist and the Guesser is computationally limited (i.e.\ cannot break them); see Goldreich~\cite{Gol01} for background on the notions. 
We claim that there is a method using a small amount of secret memory to generate a hard-to-guess permutation. We will describe two methods, one simpler, but computationally more expensive, and the other one more efficient. 

It is well known that the existence of one-way functions is equivalent to the existence of pseudorandom generators (PRGs), i.e.\ a function that maps a short seed~$s$ into a longer string where the output is indistinguishable from a truly random string function provided the seed is chosen at random. (see~Goldreich~\cite{Gol01}, Haitner et al.~\cite{HaitnerRV13} or Mazor-Pass~\cite{MazorP23}). 
Suppose that we have a PRG $G:\{0,1\}^k \mapsto \{0,1\}^{2k}$ and that $G$ does not require more than $O(k)$ bits of memory to evaluate the output. Then, we can construct from $G$ a PRG
$G':\{0,1\}^k \mapsto \{0,1\}^w$ mapping a short seed $s$ of length $k$ into a string of length $w$ which is polynomial in $k$ (the value of $m$ will be determined later) and where the amount of memory required of $G'$ to perform the mapping is $O(k)$ bits. To see how this could be done, think of a lopsided tree of depth $w$ where the output is in the leaves. Now, as in the famed GGM construction, we assign labels to the nodes. The root is labeled with the seed, and the left child of each internal node is labeled with the left part of the application of $G$ to the parent label, and the right child is labeled with the right part of the application of $G$.
We will interpret the generated string as a sequence of values in $[n]$ denoted with $X_1, X_2, \ldots$. The permutation generated is going to be a subsequence of these values. 

The output in the $i$th round is the next value $X_j$ that appears for the first time. 
In other words, if the values so far are $X_{j_1}, X_{j_2} \ldots X_{j_{i-1}}$, then the value of the $i$th element in the permutation is the first $X_j$ not in the set $\{ X_{j_1}, X_{j_2} \ldots X_{j_{i-1}}\}$.
This value can be found by regenerating the labels on the tree (from the seed) and seeing whether $X_j$ appeared before or not. This requires just $O(k)$ bits of memory. 

From the coupon collector's problem, the expected amount of work we need to do is $O(n^2\log^2 n)$ applications of the PRG (assuming each application gives us one block in $[n]$), since after $n \ln n$ turns we expect {\em all} values to appear and checking whether a value appeared does not take more than the length of the sequence. This means that we can set $w$ (the size of the tree) to be  $O(n\log^2 n)$, and with a very high probability, we will not need more values than this to generate the permutation. 
A more careful analysis shows that the amount of work required to generate the permutation is only $O(n^2\log n)$, since if $X_j$ appeared before, we need, in expectation, to check only at most $n$ previous elements.

\noindent
{\bf A more efficient solution:}
It is well known that the existence of one-way functions is equivalent to the existence of pseudorandom functions (PRFs), i.e.\ collections of functions that are indistinguishable from a truly random function (see~Goldreich~\cite{Gol01}). Given a PRF $F_k(x)$,  it can be 
used to construct pseudorandom permutations (PRP) (See Luby and Rackoff~\cite{LubyR88}) $P_k(x): [n] \mapsto [n]$. 

The natural algorithm is then: the Dealer generated a random key $k$ and stored it secretly. At step $i$ it outputs $P_k(i)$. The storage requirements are the key $k$ and the turn number, which should be much smaller than $n$. The result is a permutation, but is it hard to guess? 
If one looks at the Luby-Rackoff construction~\cite{LubyR88} or later ones~\cite{NaorR99}, then the distinguishing probability is of the form ${\ell^2}/2^n$ (or even $\sqrt{n}$ in the denominator), where $\ell$ is the number of queries the adversary has. In our case $\ell=n-1$, so such constructions (or at least their analysis) are useless, especially towards the end of the sequence. 
 
Instead, we need to apply the construction of pseudo-random permutations on small domains (related to format-preserving encryption) that are resilient to {\em any} number of queries. Such constructions are known and can be based on PRFs. 
The most efficient construction for small domain PRP is the ``Sometimes Recurse'' (SR) shuffle by Morris and Rogaway~\cite{MorrisR14} (a variant of the earlier swap-or-not and mix-and-cut shuffles), which for a domain of size $n$ runs in $\emph{expected}$ time of $O(\log n)$ (this is the number of applications of the PRF per call to the PRP) and worst case $O(\log^2 n)$. It is secure {\em even when the adversary queries the whole domain} (as will be the case in our setting).  
	
Therefore the size of memory needed to implement the scheme is the length of the key to a PRF and the total time needed to generate a permutation on $n$ elements is $O(n \log n)$ applications of the PRF, with a maximum of $O(\log^2 n)$ applications of the PRF per element. 

\begin{theorem}
If one-way functions exist, then there is a closed-book Dealer that can generate a permutation where any polynomial in the security parameter time Guesser can succeed in guessing $\ln n$ plus a negligible in the security parameter many values in expectation.  The amount of storage the Dealer needs is proportional to the security parameter. 
\end{theorem}

One remaining question is whether we can improve on the efficiency of the Dealer and get it down to $O(1)$ applications of the PRF per turn: 

\vspace{2mm}
\noindent
{\bf Question:} Is it possible to generate a random-looking permutation with low memory and O(1) calls to a PRF per element? 

Another question is whether we can prove that one-way functions are necessary for such a statement, which we discuss next. 
 
\paragraph{One-way functions and pseudorandomness:}

What are the consequences of assuming that one-way functions do not exist? Can we say that for any efficient (poly-time) Dealer there is an efficient Guesser that wins in expectation more than $\Omega(n/m + \log m)$? If we think of the myopic optimizing Guesser, then it may seem that we can approximate it by coming up with a {\em random inverse of the Guesser} and act similarly to the open-book case. That is given the Dealer's choices so far, find a random set of coins $\delta$ that yields its choices (i.e.\ choose uniformly from the set collections of coins that is consistent with its choices) and see what the next choice is; repeat this several times to get a good approximation for the most likely move. 

As we know from Impagliazzo and Luby~\cite{ImpagliazzoL89}, if one-way functions do not exist, then it is possible to find a random inverse of a function.  But the problem with this approach is that it works for one turn. If we want to do it again and again for each turn, then we must take into account the {\em complexity of the inverter}. These costs keep increasing as we compose the inverter more and more times. So it seems that this approach does not yield the desired result. This problem arises, for instance, also in the context of repeated games with limited randomness where the issue is where there is a computationally efficient way to exploit a player with limited randomness if one-way functions do not exist (see Section 5 of Hub{\'{a}}cek, Naor and Ullman~\cite{HubacekNU16}). We conclude with the following open problem:

\vspace{2mm}
\noindent
{\bf Question:} Is it true that that breaking the $O(n/m+\log n)$ bound efficiently is possible iff one-way functions exist?

\subsection*{Other Distributions}
We point out that in some cases there are methods with small space for the generation of a large object. Consider, for example, the generation of a $2n$ bit string with exactly $n$ ones and $n$ zeroes. How much memory is needed to output such a string on-the-fly? We claim that $O(\log n)$ bits of memory suffice. At any point $t$ in time ($1 \leq t \leq 2n$), simply store the number $k_t$ of bits up to point $t$ that were one. The distribution of the next bit $x_i$ is one with probability $\frac{n-k_t}{2n-t+1}$ and zero with probability $\frac{n+k_t-t+1}{2n-t+1}$.  

\vspace{2mm}
\noindent
{\bf Question:} Characterize the distributions where this low memory generation is possible.

\section*{Acknowledgments}
We thank Udi Wieder and Yotam Dikstein for meaningful discussions and advice. 

\bibliographystyle{alpha}
\bibliography{biblio}

\begin{appendices}
    \section{Algorithmic Description of the Perfect Dealer}

\label{sec_linear_perfect:sec_appendix}

\begin{algorithm}[]
\caption{To draw an element from a cell, sample a random $1$-bit and toggle it off.}
\label{sec_linear_perfect:alg_draw_card_from_cell}
\small
\begin{algorithmic}
\Procedure{DrawElementFromCell}{cell}
    \State
        $c$ = $\popCount$(cell.elements)
    \Comment{Count the number of elements in the cell}
    \State
        $r \sim \{0,..., c-1\}$
    \State
    $i$ = $\bitSelect$(cell.elements, $r$)
    \Comment
        Get the index of the $r$th $1$-bit
    \State
        element = $($cell.$j \cdot \log(n)) + i$
    \State
    \Return element
\EndProcedure

\end{algorithmic}
\end{algorithm}

\begin{algorithm}[]
\caption{Intervals}
\label{sec_linear_perfect:alg_draw_cell_of_population}
\small
\begin{algorithmic}

\Procedure{SampleACellOfPopulationSize}{population\_size}
    \LineComment
    Sample a random cell
    \State
    $s \gets $interval\_size[population\_size]
    \State
    $r \sim \{0,..., s - 1\}$
    \State
    $l \gets $ interval\_beginnings[population\_size] $+$ $r$
    \Comment
        location of a random cell in the interval
    \State
        \Call{DrawElementFromCell}{cells[$l$]}
\EndProcedure

\Procedure{DecrementCellPopulationSize}{$l$, population\_size}
    \LineComment
    Update intervals
    \LineComment
    1. Swap the location of $l$ with the last index in the interval
    \State
    $b$ = interval\_beginnings[population\_size] + interval\_size[population\_size] $- 1$
    \State
    cells[$l$] $\leftrightarrow$ cells[$b$]
    \\
    \LineComment
    2. Update interval boundaries
    \State
    interval\_size[population\_size] -= 1
    \State
    interval\_size[population\_size-1] += 1
    \State
    \begin{varwidth}[t]{\linewidth}
      interval\_beginnings[population\_size-1] = \par
        \hskip\algorithmicindent interval\_beginnings[population\_size] + interval\_size[population\_size]
      \end{varwidth}
\EndProcedure
\end{algorithmic}
\end{algorithm}

\begin{algorithm}[]
\caption{Urn Data Structure}
\small
\label{sec_linear_perfect:alg_urn_ds}
\begin{algorithmic}
\LineComment A single node in the linked list
\State \textbf{struct} Node $\{$
    \State \hspace{1em} color: int
    \Comment{$\log \colors$ bits}

    \State \hspace{1em} prev: int
    \Comment{$\log \marbles$ bits} 
    \State \hspace{1em} next: int
    \Comment{$\log \marbles$ bits} 
\State $\}$
\LineComment The Urn data structure
\State \textbf{struct} Urn $\{$
    \State \hspace{1em} nodes []Node
    \Comment{An array of $\marbles$ nodes}

    \State \hspace{1em} size int
    \Comment{Number of total marbles currently in the Urn}
    \State \hspace{1em} anchors [color]
    \Comment{$\colors$ anchors to different linked lists}
\State $\}$

\end{algorithmic}
\end{algorithm}

\begin{algorithm}[]
\caption{Urn Methods}
\small
\label{sec_linear_perfect:alg_urn_methods}
\begin{algorithmic}
\Procedure{AddMarble}{color}
    \State
        nodes[size].color $\gets$ color
    \State
        nodes[size].prev $\gets$ Nil
    \Comment{set the color of the last node}
    \State
        head $\gets $ anchors[color]
    \If {head $\neq$ Nil}
    \Comment{if there is already a linked list for this color}
        \State
            anchors[color].prev $\gets$ size
        \State
            nodes[size].next $\gets$ head
    \EndIf
    \State
        anchors[color] $\gets$ size
    \Comment{Set the new anchor}
\EndProcedure

\Procedure{RemoveMarble}{color}
    \State
        prev\_head $\gets$ anchors[color]
    \State
        new\_head $\gets$ nodes[prev\_head].next
    \State
        size $\gets$ size - 1
    \State
        nodes[prev\_head] $\gets$ nodes[size]
        \Comment{Copy the content of the last node}
    \LineComment
        Update neighbors of moved element
    \If {nodes[size].next $\neq$ Nil}
        \State
        nodes[nodes[prev\_head].next].prev $\gets$ prev\_head
    \EndIf
    \If {nodes[size].prev $\neq$ Nil}
        \State
        nodes[nodes[prev\_head].prev].next $\gets$ prev\_head
    \EndIf
    \LineComment
        If we moved the anchor of some linked list, update the anchor table
    \If {nodes[prev\_head].prev == Nil}
        \State
            anchors[nodes[prev\_head].color] $\gets$  prev\_head
    \EndIf 

    \State
        anchors[color] $\gets$ new\_head
    \If {new\_head $\neq$ Nil}
        \State 
            nodes[new\_head].prev $\gets$ Nil
    \EndIf
\EndProcedure

\Procedure{SampleMarble}{}
    \State
        i $\sim$ [size]
    \Comment{sample a Marble uniformly at random}
    \State \Return nodes[i].color
\EndProcedure

\end{algorithmic}
\end{algorithm}

\end{appendices}

\end{document}